\documentclass[12pt,a4paper]{article}
\pdfoutput=1

\usepackage[utf8]{inputenc}
\usepackage[T1]{fontenc}
\usepackage{lmodern}
\usepackage[margin=1in]{geometry}
\usepackage{amsmath,amssymb,amsthm}
\usepackage{graphicx}
\usepackage{booktabs}
\usepackage{enumitem}
\usepackage[hyperfootnotes=false]{hyperref}
\usepackage{natbib}
\usepackage{setspace}
\usepackage{caption}
\usepackage{float}

\newtheorem{proposition}{Proposition}
\newtheorem{theorem}{Theorem}
\newtheorem{lemma}{Lemma}
\newtheorem{corollary}{Corollary}
\theoremstyle{remark}
\newtheorem*{remark}{Remark}

\title{AI as Coordination-Compressing Capital:\\
  Task Reallocation, Organizational Redesign,\\
  and the Regime Fork\thanks{This version supersedes arXiv:2602.16078v3. It corrects an error in v3's Proposition~4, which stated that managerial inequality rises faster with agent capital when $\beta$ is larger (this does not hold in general), and shows that a positive coordination floor changes the long-run distributional results, contrary to a statement in v3's limitations section. Appendix~\ref{app:changes} lists all changes.}}
\author{Alex Farach\thanks{Email: \href{mailto:alexfarach@microsoft.com}{alexfarach@microsoft.com}. The views expressed in this paper are the author's own and do not necessarily reflect those of Microsoft. The author acknowledges the use of AI assistants for coding, simulation implementation, and draft preparation. All theoretical contributions, modeling decisions, and interpretive claims are solely the author's.}\\ {\normalsize Microsoft}}
\date{First version: February 2026. This version: September 2026.}

\begin{document}
\maketitle

\begin{abstract}
Task-based models of AI hold organizational structure fixed. We model AI as agent capital that compresses managers' per-link coordination costs toward residual floors for consequential decisions. Positive floors bound spans of control; zero floors do not. For a fixed manager pool, proportional team sizes, equal mean team quality and fixed within-firm pay shares, and when every floor is positive, the floors determine long-run managerial wage dispersion, independent of execution-compression parameters. With a common positive floor, distinct skills and skill-biased execution savings, managerial inequality is purely transitional: zero at zero capital and in the long run, positive at every positive finite capital level. If every lower-floor manager approaches capacity faster than every higher-floor one, inequality exceeds its long-run level at sufficiently large finite capital and converges from above, a non-monotone path. Numerical illustrations, not estimates, include a simulation with skill-sorted workers, to which the wage and Gini results do not apply.
\end{abstract}

\medskip
\noindent\textbf{JEL Codes:} D23, J24, L23, O33

\noindent\textbf{Keywords:} Artificial intelligence, coordination costs, organizational design, span of control, task-based models, inequality

\section{Introduction}

Firms that adopt AI appear to become flatter. \citet{ewens2025} document that U.S.\ public firms reduced management layers following AI adoption, in a sample of more than 3{,}100 firms, and \citet{babina2023} find that firms investing in AI shifted toward flatter structures with proportionally fewer mid- and senior-level employees. The leading framework for AI's labor-market effects is task-based. Technology changes which tasks are feasible, who performs them and at what cost \citep{acemoglu2019,agrawal2026}. These models hold fixed the hierarchy through which work is coordinated. Workers move across tasks, while spans of control, management layers and reporting lines stay exogenous.

This paper asks what happens to spans of control and managerial pay when AI compresses routine coordination but some consequential decisions still need a person. The classical answer to the first half of the question is that cheaper communication widens spans \citep{garicano2006}. We ask how that comparative static changes when the per-link coordination cost has a positive lower bound.

In the model, each manager supervises a team, and each worker link uses some of the manager's time. Agent capital $K_A$, the firm's stock of coordination-compressing AI capability, lowers the routine, compressible part of that per-link coordination cost, and the savings may be larger for more skilled managers, with strength $\beta$. A residual part $f_i$ of the per-link coordination cost, the \emph{fork-decision floor}, remains. It stands for consequential decisions at forks (points where work can proceed in materially different directions) whose alternatives cannot cheaply be explored before commitment. In a newsroom, for example, AI tools can draft routine briefs, but an editor still has to decide whether to run a contested story. With every floor at zero, the model reduces to the case in which cheaper coordination widens spans without limit.

We obtain three main results.
\begin{enumerate}[leftmargin=*,itemsep=2pt]
  \item \emph{A positive floor bounds the span.} Each manager's span of control rises with agent capital toward $1/f_i$ rather than growing without bound, and continuous manager demand with a common positive floor falls toward a positive limit (Propositions~\ref{prop:span} and~\ref{prop:manager-demand}). The classical prediction that cheaper coordination keeps widening spans without limit holds only at a zero floor.
  \item \emph{The floors set long-run managerial pay dispersion.} For a fixed manager pool with proportional team sizes, equal average team quality and fixed within-firm pay shares, and when every floor is positive, the long-run managerial Gini depends only on the floors and the labor share, not on $\beta$ or on the compression rate $\gamma$ (Proposition~\ref{prop:wage-dispersion}). A common positive floor makes managerial inequality purely transitional (Corollary~\ref{cor:common}).
  \item \emph{Inequality can overshoot on the way.} If every lower-floor manager also approaches capacity faster than every higher-floor manager (condition C), managerial Gini rises above its long-run value at sufficiently large finite capital and converges to it from above (Theorem~\ref{thm:overshoot}). Under the reverse ordering of a power floor, it stays below its limit at every finite capital level (Corollary~\ref{cor:reverse}). In the zero-floor model with $\beta>0$ and distinct skills, managerial Gini instead rises throughout.
\end{enumerate}
The floors $f_i$ and their skill gradient are imposed, not measured, and the numerical exercises illustrate the propositions for one calibration; they are not evidence about firms. Whether condition C, or any particular floor distribution, holds in firms is unknown.

The paper makes three contributions. First, it places a boundary on the classical communication-cost comparative static. The inverse-cost span mapping is standard, and the distinction between prediction and residual judgment comes from \citet{agrawaljudgment2019,agrawalcomplexity2019}; neither is original to this paper. What we add is a positive lower bound on the per-link coordination cost as a function of agent capital, which nests the zero-floor case. Second, it shows how that bound changes the distribution of managerial pay. \citet{garicano2006} show that cheaper communication can concentrate returns at the top. In our model, under the maintained allocation and pay rules and when every floor is positive, the floors set the long-run level of that concentration, while $\beta$ and $\gamma$ shape only the path. Third, it gives a sufficient condition for a non-monotone path of managerial inequality, together with a reverse-ordering counterpart. \citet{rosen1981} provides a related superstar scaling argument, rather than evidence for this floor.

Section~\ref{sec:model} sets out the model and its standing assumptions, and Sections~\ref{sec:capacity}, \ref{sec:longrun}, \ref{sec:transition} and~\ref{sec:fork} derive spans and manager demand, long-run managerial inequality, the transition and overshoot, and employment and the regime fork. Section~\ref{sec:numerical} gives the numerical illustration, Section~\ref{sec:discussion} discusses predictions, policy and limitations (Section~\ref{sec:limitations}), Section~\ref{sec:conclusion} concludes, and Appendix~\ref{app:proofs} contains all proofs.

\subsection{Related Literature}

\paragraph{Span of control and communication costs.} Supervision and communication costs have long constrained organizational size \citep{williamson1967,calvo1978,keren1979,lucas1978,rosen1982,radner1993}. In knowledge hierarchies \citep{garicano2000,garicano2006}, cheaper communication expands spans and can concentrate returns at the top, while cheaper knowledge acquisition has different distributional implications. Our $\beta$ governs the skill gradient of \emph{execution} compression and its transitional effect, not which of these forces dominates in the long run when floors are positive.

\paragraph{Prediction, judgment and AI in organizations.} The distinction between prediction and residual judgment is due to \citet{agrawaljudgment2019,agrawalcomplexity2019}. Recent models let organizations adjust to AI. \citet{ide2025} vary the knowledge hierarchy under AI while retaining an exogenous communication-cost parameter, and \citet{xu2025} find that span initially contracts before eventually expanding, also without making that parameter an explicit function of AI capital. \citet{cheng2025} model AI replacement in teams with sequential information flow, showing that positional structure, not only skill level, determines displacement risk. Our question is narrower: how a positive lower bound on the per-link coordination cost alters the communication-cost comparative static.

\paragraph{Task-based models.} \citet{acemoglu2019} formalize a displacement-reinstatement dynamic, in which automation removes workers from existing tasks while new task creation reinstates labor demand. \citet{agrawal2026} show that AI can enhance worker productivity without automating tasks, by expanding the share of tasks that workers with AI expertise can perform ($\bar{\tau}$) or broadening individual workers' feasible task sets ($\Omega(i)$). They show that aggregate productivity depends on the geometric-mean uniformity $\exp(g)$ of equilibrium skill densities across tasks, while wage inequality depends inversely on the Shannon differential entropy of the same densities. The two typically move in opposite directions, but not always: an extensive-margin improvement in AI can raise both. Their task continuum with comparative advantage differs from our manager-team pairs with coordination friction, but both hold organizational architecture exogenous. We add an organizational channel. At $K_A=0$ the hierarchy is fixed and only the task-reallocation channel operates, through $\bar{\tau}$ or $\Omega(i)$. At $K_A>0$ coordination compression adds a structural channel, and geometric-mean uniformity and entropy remain valid summary statistics that become endogenous to organizational structure rather than determined solely by worker reallocation within a fixed architecture.

\paragraph{Evidence on technology and organization.} Beyond the flattening evidence in the paper's first paragraph, \citet{bloom2014} show that information technology and communication technology have distinct and sometimes opposing effects on firm organization (information technology decentralizes, communication technology centralizes), which suggests that the organizational channel deserves explicit modeling. \citet{bresnahan2002} find that IT investment is associated with larger productivity gains when combined with complementary organizational restructuring, which does not establish that reorganization must dominate direct task effects. Standard growth accounting absorbs AI's contribution into TFP, obscuring its role as a distinct input rather than an augmentation of existing capital or labor \citep{farach2025}.

\section{Model}\label{sec:model}

\subsection{Firms, Teams and Production}\label{sec:production}

Consider a hierarchical firm with production workers and managers. To fix ideas, aggregate output takes the familiar Cobb--Douglas form $Y = A \cdot L_{\text{eff}}^\alpha \cdot M^{1-\alpha}$, where $L_{\text{eff}}$ is effective labor, $M$ is management, and $\alpha \in (0,1)$ is the labor share. We use this aggregate expression as motivation; the formal model operates at the \emph{team level}.

Each manager $i$ oversees a team of $n_i$ workers. Team output is
\begin{equation}\label{eq:production}
  Y_i = A \cdot L_{\text{eff},i}^\alpha \cdot m_i^{1-\alpha}, \qquad m_i = 1,
\end{equation}
where $A > 0$ is total factor productivity and $m_i = 1$ reflects that each team has exactly one manager. Team output therefore simplifies to $Y_i = A \cdot L_{\text{eff},i}^\alpha$, and aggregate output is $Y = \sum_i Y_i$. The two-input structure of~\eqref{eq:production} motivates the pay rule below.

Effective labor is raw labor discounted by coordination friction. Workers are heterogeneous, with individual skills $q_j$. Effective labor on manager $i$'s team is
\begin{equation}\label{eq:effective-labor}
  L_{\text{eff},i} = \frac{Q_i}{1 + c_i \cdot n_i},
\end{equation}
where $Q_i = \sum_{j \in \text{team}_i} q_j$ is team quality, $n_i$ is headcount, and $c_i$ is the per-link coordination cost defined in Section~\ref{sec:cost}. As teams grow, the manager's total coordination cost $c_in_i$ rises: each additional worker requires managerial attention proportional to $c_i$. This captures a classical result in models of large teams facing coordination losses~\citep{bolton1994,radner1993}. Quality enters through the numerator, and the homogeneous-worker case sets $Q_i=n_i$.

Pay follows a Cobb--Douglas sharing rule. The manager receives $(1-\alpha)Y_i$ as a residual claimant, and the worker share $\alpha Y_i$ is divided among team members in proportion to their skill contribution, so worker $j$ on team $i$ receives $(q_j / Q_i) \cdot \alpha \cdot Y_i$. Since $m_i = 1$ is a discrete constraint (not a variable input), this is a rent-sharing convention, not competitive factor pricing; we adopt it for tractability and because the $(1-\alpha, \alpha)$ split is the natural division under Cobb--Douglas structure.

\paragraph{Notation.} Throughout, $s_i$ denotes managerial skill and $q_j$ denotes worker skill. The index $i$ refers to managers, $j$ to workers. Table~\ref{tab:variables} collects the notation.

\subsection{Coordination Cost with a Residual Floor}\label{sec:cost}

Agent capital reduces routine execution overhead, while some consequential fork decisions may still consume managerial attention. We start with identical managers. The per-link coordination cost is
\begin{equation}\label{eq:coordination-cost}
  c(K_A) = f + \frac{c_0-f}{1+\gamma K_A}, \qquad 0\leq f<c_0,
\end{equation}
where $c_0>0$ is the baseline per-link coordination cost, $\gamma>0$ controls execution compression, and $f$ is its residual floor. At zero capital $c=c_0$. As capital grows, the compressible part $c_0-f$ shrinks and $c$ approaches the floor $f$. With $f=0$, cost falls toward zero.

Managers differ in how much they gain from agent capital. We model this through an elite complementarity parameter $\beta \geq 0$. Manager $i$ with skill $s_i \in (0, 1]$ has effective agent capital
\begin{equation}\label{eq:elite}
  K_{\text{eff},i} = K_A \cdot s_i^\beta .
\end{equation}
When $\beta$ is low (e.g., 0.2), $s_i^\beta$ is closer to one across managers, and execution savings are broadly distributed. When $\beta$ is high (e.g., 3.0), execution savings initially favor high-skill managers. Neither case determines the distribution of residual fork-decision costs $f_i$, and the power function alone cannot sustain indefinite elite span growth under a positive common floor. Substituting~\eqref{eq:elite} into~\eqref{eq:coordination-cost} and allowing manager-specific floors gives
\begin{equation}\label{eq:ci}
  c_i(K_A) = f_i + \frac{c_0-f_i}{1 + \gamma \cdot K_A \cdot s_i^\beta}, \qquad f_i\in[0,c_0).
\end{equation}
Equation~\eqref{eq:coordination-cost} is the case $s_i=1$ and $f_i=f$. We write $a_i=\gamma s_i^\beta$. At zero capital $c_i=c_0$; at high capital $c_i$ tends to $f_i$. Setting every $f_i=0$ gives exactly $c_i=c_0/(1+\gamma K_A s_i^\beta)$ and $S_i=(1+\gamma K_A s_i^\beta)/c_0$. We model $K_A$ separately from $K$ or $L$ because digital labor's scalability and context-dependent substitutability~\citep{farach2025} make the organizational margin worth identifying explicitly.

How $f_i$ varies across managers is an empirical question. A common positive floor is a special case, not the main assumption. As an illustration we use the \emph{power floor} $f_i=f_0s_i^{-\theta}$, with $\theta\geq0$ and $0<f_0s_{\min}^{-\theta}<c_0$. It assumes perfect rank alignment of execution-complementary skill and residual fork efficiency; a separate latent decision index need not share that ranking. An informational floor common across managers corresponds approximately to $\theta=0$, while $\theta>0$ may represent judgment only after other skill-correlated costs are excluded.

\subsection{Capacity, Layers and the Task Frontier}

The coordination cost implies a \emph{coordination capacity},
\begin{equation}\label{eq:span}
  S_i(K_A) = \frac{1}{c_i(K_A)},
\end{equation}
the maximum team size manager $i$ can effectively coordinate, given the technology of coordination. This is a technological capacity constraint, not the solution to a manager's optimization problem. We impose $n_i \leq S_i$; when worker supply or the frontier binds, realized teams can be smaller. At binding capacity, $c_i n_i=1$ and $L_{\text{eff},i}=Q_i/2$.\footnote{Each link consumes $c_i$ units of a manager's unit time endowment. Thus $c_i n_i\leq1$ implies $n_i\leq1/c_i$. A positive residual bounds an incumbent's capacity, while the firm could add managers at a cost outside the fixed-pool closure of Section~\ref{sec:assumptions}.} For identical managers, $S=(1+\gamma K_A)/(c_0+f\gamma K_A)$, which tends to $1/f$ if $f>0$ and is linear in $K_A$ if $f=0$.

For identical managers with $S>1$, the number of hierarchical layers required for $N$ workers is
\begin{equation}\label{eq:layers}
  \text{Layers}(N, S) = \lceil \log(N) / \log(S) \rceil .
\end{equation}
As $S$ rises, layers compress. This is \emph{hierarchy substitution}, distinct from task substitution.

Coordination compression may also expand the feasible set of organizational tasks. In a newsroom, activities that were previously too costly to coordinate, such as newsletter verticals, podcast production, interactive graphics and real-time analytics, become viable as coordination friction falls. The empirical relevance of endogenous task creation is well documented: \citet{autor2024} show that the majority of contemporary employment is in job categories that did not exist in 1940, with new work shifting from middle-skill to high-skill occupations after 1980. We model the task frontier as
\begin{equation}\label{eq:frontier}
  T(K_A) = T_0 \cdot (1 + \delta \cdot K_A),
\end{equation}
where $T_0$ is the baseline task frontier and $\delta \geq 0$ controls task creation.\footnote{Linearity is a convenience. Any weakly increasing frontier together with weakly increasing $D$ yields weakly increasing $E=\min(D,T,N)$, but employment can stop rising when residual supervisory capacity binds.} With $\delta=0$ the frontier is fixed; with $\delta>0$ feasible tasks expand, but some may remain unstaffed. Employment is $E(K_A)=\min(D(K_A),T(K_A),N)$, where $D=\sum_i S_i$ is total capacity and $N$ is the workforce.

\subsection{Standing Assumptions}\label{sec:assumptions}

Three maintained assumptions carry the organizational and distributional results.

\begin{description}[leftmargin=0pt,itemsep=3pt]
  \item[FP (fixed manager pool).] The set of potential managers is fixed. The analysis is short run, with no managerial entry, hiring or promotion.
  \item[PA (proportional allocation).] Team sizes are proportional to capacity, $n_i=\lambda S_i$, with $\lambda>0$ common to all managers, and every team has the same positive mean worker quality $q$, so that $Q_i=qn_i$. This rule is maintained rather than derived from optimization; optimizing worker assignment is not analyzed here.
  \item[RS (rent sharing).] Pay follows the within-firm Cobb--Douglas rule of Section~\ref{sec:production}, under which the manager receives $(1-\alpha)Y_i$, and worker $j$ on team $i$ receives $(q_j/Q_i)\alpha Y_i$.
\end{description}

\noindent Throughout, Propositions~\ref{prop:output} and~\ref{prop:span} need none of these assumptions. Propositions~\ref{prop:manager-demand} and~\ref{prop:frontier} maintain FP. Proposition~\ref{prop:wage-dispersion}, Corollaries~\ref{cor:common} and~\ref{cor:reverse}, and Theorem~\ref{thm:overshoot} maintain FP, PA and RS. The sorted-worker simulation of Section~\ref{sec:numerical}, which rounds team sizes to whole numbers and matches workers to managers by positive assortative matching (PAM), violates PA, so results that rely on PA do not apply to it.

\subsection{Discussion of the Model}

\paragraph{What the floor represents.} The split between compressible execution and a residual floor is a hypothesis about attention, not an estimate from a benchmark. The floor may reflect common informational uncertainty (a $\theta$ near zero), differential judgment ($\theta>0$), or accountability; the model does not distinguish them. In the benchmark of \citet{pan2026}, mean model accuracy falls from 62.3\% on forks whose deciding evidence already appears in the run before the fork to 21.0\% on forks requiring substantial later work, while a two-model reasoning-budget comparison shows little overall change. Neither result identifies $f_i$ or its skill gradient $\theta$.

\paragraph{Which decisions belong in the floor.} The restriction to consequential forks that cannot cheaply be explored before commitment draws on sequential search \citep{weitzman1979} and irreversible investment \citep{dixit1994}. Cheap parallel trials can instead convert judgment into compressible computation. An empirical classification must specify irreversibility, cost, or delay \emph{before} examining managerial spans.

\paragraph{Per-link versus per-manager costs.} Both routine execution overhead and residual fork-decision attention enter the manager's total coordination cost $c_i n_i$; this linear scaling with headcount is a maintained assumption. A project-level fixed obligation instead consumes a fixed amount of managerial time, but cannot by itself bound spans (Proposition~\ref{prop:span} and Appendix~\ref{app:extensions}).

\paragraph{Allocation and pay rules.} RS is a within-firm convention, and PA is a tractable benchmark that makes coordination losses equal across teams. The sorted-worker simulation instead uses PAM, motivated by knowledge-hierarchy models \citep{garicano2006} but imposed rather than derived from a competitive assignment game here.

\paragraph{Human managers and AI agents.} We present the model in terms of human managers coordinating human workers \citep{garicano2000,garicano2006}, because that case connects to existing organizational evidence \citep{ewens2025}. A worker who supervises several AI agents may also face monitoring and integration costs for each agent. Routine monitoring may scale while consequential fork judgments still consume attention. Treating that as the same $c_i$ requires a separate test of the per-link assumption. The analogy does not imply that the worker's measured skill maps to $s_i$, or that an agent's benchmark accuracy determines a human supervisor's floor.

\begin{table}[ht]
  \centering
  \caption{Notation}\label{tab:variables}
  \small
  \begin{tabular}{@{}ll@{}}
    \toprule
    Symbol & Definition \\
    \midrule
    $K_A$ & Agent capital (coordination-compressing AI capability) \\
    $c_i(K_A)$ & Per-link coordination cost: $f_i+(c_0-f_i)/(1+\gamma K_A s_i^\beta)$ \\
    $f_i$ & Fork-decision floor: residual per-link coordination cost; $0\leq f_i<c_0$ \\
    $s_i$, $q_j$ & Manager skill, worker skill \\
    $n_i$, $Q_i$ & Team headcount, team quality $\sum_j q_j$ \\
    $\gamma$ & Execution compression effectiveness \\
    $\beta$ & Skill complementarity in compressible execution \\
    $\theta$ & Skill gradient of the power floor $f_0s_i^{-\theta}$ \\
    $\delta$ & Linear task-frontier slope relative to $T_0$ \\
    $\alpha$ & Labor share in team production \\
    $S_i(K_A)$ & Coordination capacity: $1 / c_i(K_A)$ \\
    $T(K_A)$ & Task frontier: $T_0 \cdot (1 + \delta \cdot K_A)$ \\
    $h_i$, $t_i$ & Convergence index $(c_0-f_i)/(f_ia_i)$; progress $f_iS_i$ \\
    \bottomrule
  \end{tabular}
\end{table}

\section{Capacity and Hierarchy}\label{sec:capacity}

Start with identical managers. Each worker link costs $c(K_A)$ of a manager's time, so capacity is $S=(1+\gamma K_A)/(c_0+f\gamma K_A)$. With $f=0$, agent capital divides the whole cost and capacity rises linearly without bound. With $f>0$, only the compressible part $c_0-f$ shrinks, and capacity approaches $1/f$. This section states the results for output, spans and manager demand that follow from this observation. All proofs are in Appendix~\ref{app:proofs}.

\begin{proposition}[Output Effect]\label{prop:output}
  If $\gamma>0$, $0\leq f_i<c_0$, and at least one team has positive quality and headcount, then holding allocations fixed, $\partial Y/\partial K_A>0$.
\end{proposition}

Holding every team fixed, agent capital lowers each manager's per-link coordination cost, raises effective labor on every nonempty team, and so raises output. Each fixed team's output converges to $A[Q_i/(1+f_in_i)]^\alpha$, a bounded gain when floors are positive. The fixed-team sign holds regardless of $\beta$ and $\delta$. This organizational effect is distinct from task-level productivity enhancement in \citet{agrawal2026}. In the sorted-worker simulation, realized output can fall locally as capital rises (Appendix~\ref{app:extensions}).

\begin{proposition}[Span Expansion]\label{prop:span}
  For $\gamma>0$ and $0\leq f_i<c_0$, each manager's capacity increases in $K_A$. It has limit $1/f_i$ if $f_i>0$ and grows linearly without bound if $f_i=0$.
\end{proposition}

Capacity is $S_i=(1+a_iK_A)/(c_0+f_ia_iK_A)$. Agent capital raises the numerator and, when $f_i>0$, the denominator at rates whose ratio tends to $1/f_i$. The classical prediction that cheaper communication widens spans \citep{garicano2006} therefore holds here, but with a positive floor it holds only up to a ceiling. The derivative with respect to skill is less simple. For a common positive floor, the cross-partial of capacity in $a_i$ and $K_A$ has the sign of $c_0-fa_iK_A$, so it can change sign; with unequal floors it also depends on how $f_i$ varies with skill. A fixed project-level obligation per manager does not by itself bound spans (Appendix~\ref{app:extensions}).

Propositions~\ref{prop:output} and~\ref{prop:span} imply bounded gains in fixed-team output and bounded supervisory capacity for managers with positive floors; with a zero floor, capacity is unbounded. Figure~\ref{fig:floor-span} shows the span boundary for the highest-skill manager.

\begin{figure}[ht]
  \centering
  \includegraphics[width=\textwidth]{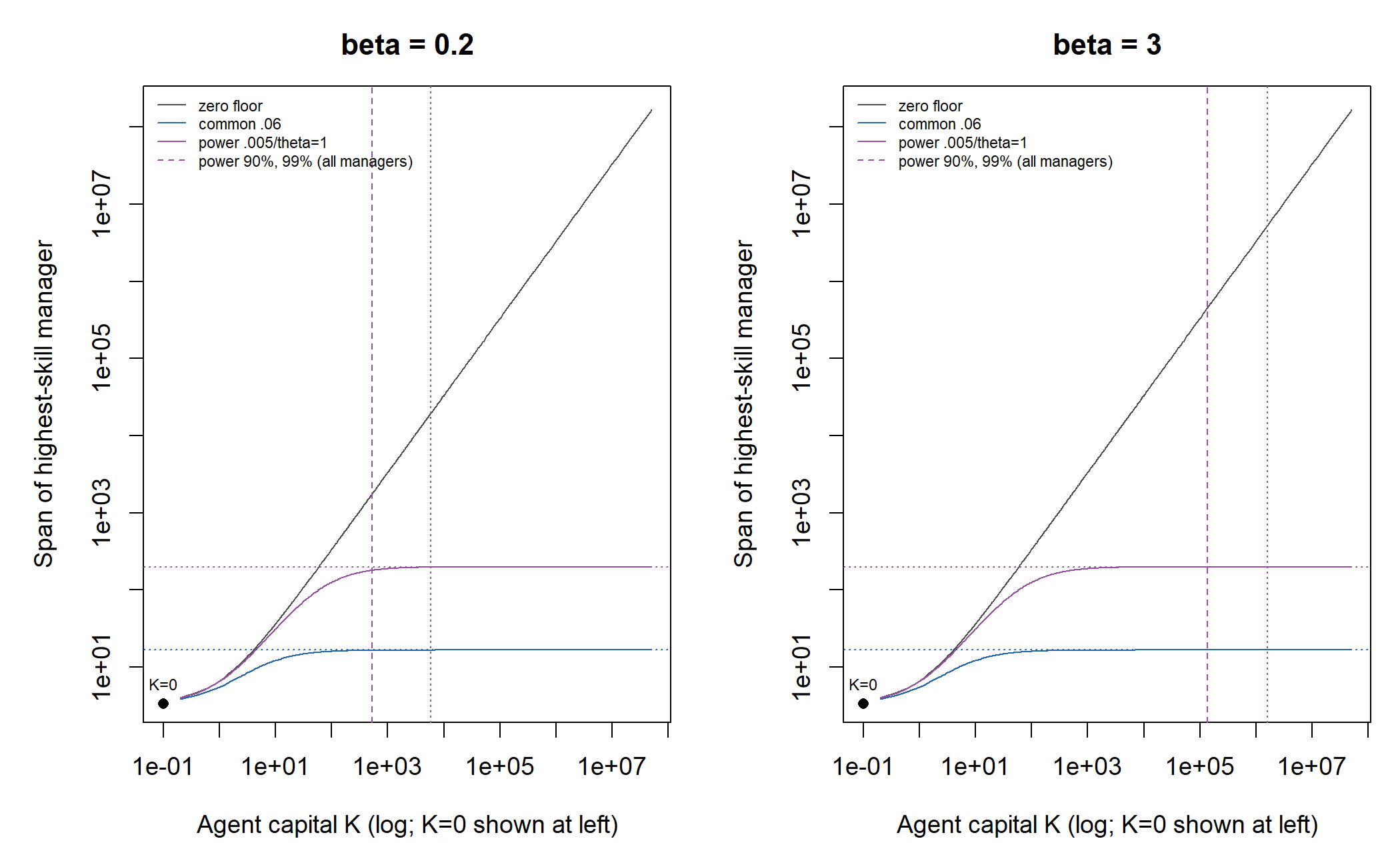}
  \caption{Span of the highest-skill manager ($s=1$) under the cost with a floor, for $\beta=0.2$ (left) and $\beta=3$ (right), with $c_0=0.3$ and $\gamma=1$. Gray: zero floor, which grows without bound. Blue: common floor $f=0.06$. Purple: power floor $f_0s^{-\theta}$ with $(f_0,\theta)=(0.005,1)$. Horizontal dotted lines mark the limits $1/f_i$. Vertical dashed and dotted guides mark the capital at which \emph{every} manager in the 20-manager pool (skills evenly spaced on $[0.05,1]$) reaches 90\% and 99\% of their own limiting span under the power floor. Both axes are logarithmic, so the $K_A=0$ baseline is marked separately at the left.}\label{fig:floor-span}
\end{figure}

\begin{proposition}[Manager Demand]\label{prop:manager-demand}
  For a fixed workforce and eligible manager pool, the minimal number of managers covering it is weakly decreasing in $K_A$ if coverage is feasible. For common $0\leq f<c_0<1$ and sufficiently large $N$, continuous $M\approx N/(S-1)$ strictly decreases; its limit is positive for $f>0$ and zero for $f=0$.
\end{proposition}

Because every manager's capacity rises, any set of managers that can cover the workforce at one level of capital can still cover it at a higher level, so the minimal count cannot rise. With a common floor, the continuous count falls toward $N/(1/f-1)$, which is positive when $f>0$. If all floors are positive, total capacity tends to $\sum_i1/f_i$ and manager $i$'s long-run capacity share is $(1/f_i)/\sum_j(1/f_j)$.\footnote{Whole-number manager and layer counts can remain flat before falling, including at zero floors (Appendix~\ref{app:changes}). The fixed-manager simulation does not model displacement.}

\section{Long-Run Managerial Inequality}\label{sec:longrun}

We now turn to pay. Under PA and RS, a manager's wage is proportional to capacity raised to the labor share, so pay dispersion among managers is a statement about the dispersion of spans.

\begin{proposition}[Wage Dispersion]\label{prop:wage-dispersion}
  Fix a pool of $n\geq2$ managers within a firm with positive proportional worker allocation, equal positive mean team quality, and the stated rent-sharing convention. For $0\leq f_i<c_0$, $w_{M,i}/w_{M,j}=(S_i/S_j)^\alpha$. If all $f_i>0$, long-run ratios equal $(f_j/f_i)^\alpha$ and managerial Gini tends to $G(\{f_i^{-\alpha}\})$, independent of $\beta$ and $\gamma$. A common positive floor gives zero long-run Gini. If every $f_i=0$, $\beta>0$ and at least two skills differ, the wage ratio of any higher- to lower-skill manager and the managerial Gini increase in $K_A$.
\end{proposition}

The proposition maintains FP, PA and RS (Section~\ref{sec:assumptions}).\footnote{When every $f_i=0$, the cost, span and wage-ratio expressions reduce exactly to those of an earlier version of this paper. That version also stated that managerial inequality rises faster with agent capital when $\beta$ is larger; this does not hold in general, and Appendix~\ref{app:changes} gives a two-manager counterexample. Proposition~\ref{prop:wage-dispersion} makes no claim about the rate.} The idea of the proof is short. Under PA every team has the same total coordination cost $c_in_i=\lambda$, so the only manager-specific term in pay is capacity, and $w_{M,i}\propto S_i^\alpha$. With positive floors, Proposition~\ref{prop:span} sends each $S_i$ to $1/f_i$, which gives the long-run ratios and the limiting Gini. Neither limit contains $\beta$ or $\gamma$, which govern how fast each manager approaches capacity, not where capacity ends up. With zero floors, capacity is linear in $K_A s_i^\beta$, higher-ranked managers grow proportionally faster, and the Gini rises (Lemma~\ref{lem:lorenz} in Appendix~\ref{app:proofs} supplies the Lorenz step). If only some floors are zero, those spans diverge and the positive-floor limit does not apply.

The case of one common positive floor makes the role of the floors plain.

\begin{corollary}[Common floor]\label{cor:common}
  Let every manager have the same floor $f\in(0,c_0)$, let $\beta>0$, and let skills be distinct. Under FP, PA and RS, $G_\infty=0$ and $G(K_A)>0$ at every positive finite $K_A$.
\end{corollary}

With a common floor, higher-skill managers approach the common ceiling $1/f$ faster, so spans differ at every positive finite capital level and coincide in the limit. Managerial inequality is purely transitional. It is zero at zero capital, positive in between, and zero in the long run.

When floors differ, they carry the long-run ranking. Under the power floor $f_i=f_0s_i^{-\theta}$, with $\theta\geq0$ and $0<f_0s_{\min}^{-\theta}<c_0$, the limiting wage ratio is $(s_i/s_j)^{\alpha\theta}$; this assumes perfect alignment of execution skill and residual fork efficiency. In the zero-floor simulation, managerial inequality is larger under high $\beta$ over the plotted capital range. With heterogeneous positive floors, $\beta$ and $\gamma$ affect the transition while the floor distribution determines long-run managerial dispersion. Economy-wide inequality and manager--worker gaps require separate analysis (Appendix~\ref{app:zero-figures}).

\section{The Transition and Overshoot}\label{sec:transition}

Proposition~\ref{prop:wage-dispersion} fixes where managerial inequality ends up when every floor is positive. This section asks how it gets there. Two statistics organize the answer. For manager $i$ with a positive floor, write
\[
  h_i=\frac{c_0-f_i}{f_ia_i},\qquad t_i=f_iS_i,\qquad a_i=\gamma s_i^\beta .
\]
The ratio $t_i\in[f_i/c_0,1)$ is the share of long-run capacity the manager has reached. A line of algebra gives
\begin{equation}\label{eq:progress}
  1-t_i=\frac{h_i}{h_i+1/a_i+K_A},
\end{equation}
so at high capital the remaining gap to capacity is roughly $h_i/K_A$. The index $h_i$ therefore measures how slowly manager $i$ approaches capacity. It is large when much of the initial cost is compressible ($c_0-f_i$ large) relative to how quickly agent capital removes it ($f_ia_i$).

Under PA and RS, pay can be split into a long-run base and a progress term,
\[
  w_{M,i}\propto S_i^\alpha=v_i\,\tau_i,\qquad v_i=f_i^{-\alpha},\qquad \tau_i=t_i^\alpha .
\]
The base $v_i$ is fixed by the floor, and $\tau_i$ rises toward one. A lower floor means a higher long-run base. Whether the lower-floor manager is also further along at a given capital level depends on $h$.

\begin{lemma}[Exact pairwise crossing]\label{lem:crossing}
  Let $0<f_i<f_j<c_0$. If $h_i<h_j$, then $K^*_{ij}>0$ and $S_i/S_j>f_j/f_i$ if and only if $K_A>K^*_{ij}$, where
  \[
   K^*_{ij}=\frac{c_0(f_j-f_i)}{f_i a_i(c_0-f_j)-f_j a_j(c_0-f_i)} .
  \]
  The denominator of $K^*_{ij}$ is positive exactly when $h_i<h_j$. If $h_i\geq h_j$, then $S_i/S_j<f_j/f_i$ at every finite $K_A\geq0$. In the power-floor illustration $f(s)=f_0s^{-\theta}$, $h(s)=[(c_0/f_0)s^{\theta-\beta}-s^{-\beta}]/\gamma$ and
  \[
    \frac{dh}{ds}=\frac{s^{-\beta-1}}{\gamma}
    \left[\frac{c_0}{f_0}(\theta-\beta)s^\theta+\beta\right],
  \]
  which is negative exactly when $\beta[1-f(s)/c_0]>\theta$ and positive exactly when $\beta[1-f(s)/c_0]<\theta$. For $\theta>0$, $\beta[1-f(s_{\min})/c_0]>\theta$ makes $h$ strictly decreasing on the skill support, which is the all-pairs ordering used in Theorem~\ref{thm:overshoot}.
\end{lemma}

The pair's wage ratio passes its own long-run value when the lower-floor manager, who has the higher base, is also closer to capacity. The factor $1-f/c_0$ is the initially compressible share. One qualifying pair does not establish a population-Gini overshoot. Theorem~\ref{thm:overshoot} assumes the ordering for every pair.

\begin{theorem}[Transitional overshoot under condition C]\label{thm:overshoot}
  Suppose all $n\geq2$ managers have positive floors below $c_0$, at least two floors differ, and $h_i<h_j$ for \emph{every} pair with $f_i<f_j$ (condition C). Under FP, PA and RS, managerial Gini satisfies $G(0)=0<G_\infty$, where $G_\infty=G(\{f_i^{-\alpha}\})$, and $G(K_A)>G_\infty$ for every $K_A>\bar K=\max_{f_i<f_j}K^*_{ij}$. It converges from above and is non-monotone.
\end{theorem}

The idea of the proof is as follows. Beyond $\bar K$, a lower floor gives a manager both a higher base $v_i$ and a higher progress term $\tau_i$, so the two factors of pay are ordered the same way. Multiplying a distribution by a factor that rises with rank moves Lorenz shares toward the top (Lemma~\ref{lem:lorenz}), so the Gini of $v\tau$ exceeds the Gini of $v$ alone, which is $G_\infty$. Equal spans at zero capital give $G(0)=0$. The power-floor condition of Lemma~\ref{lem:crossing} satisfies C when $\theta>0$ and skills differ.

The theorem is limited in three ways. It does not apply to PAM or integer assignment, which violate PA. Condition C is sufficient, and the theorem makes no claim when only some distinct-floor pairs qualify. It proves neither unimodality nor a rate, peak location, or magnitude. Appendix~\ref{app:numerical} gives two numerical examples in which some pairs qualify but C fails: an analytic three-manager example and a constructed 20-manager population.

The reverse ordering has a counterpart for the power floor.

\begin{corollary}[Reverse power-floor ordering]\label{cor:reverse}
  For $\theta>0$, distinct skills and $f(s)=f_0s^{-\theta}<c_0$, if $h(s)$ is strictly increasing throughout the skill support, then under FP, PA and RS, $G(K_A)<G_\infty$ at every finite $K_A\geq0$. The strict support-wide condition $\beta[1-f(s_{\max})/c_0]<\theta$ suffices.
\end{corollary}

Here the lower-floor manager has the higher base but is always further from capacity, so the two factors of pay pull in opposite directions and the Gini approaches its limit from below. Reversed $h$ ordering alone does not prove this for arbitrary floors if span ranks can reverse. Mixed $h$ orderings carry no general no-overshoot conclusion. Figure~\ref{fig:floor-gini} shows one path of each kind, computed in the proportional-allocation exercise; Section~\ref{sec:numerical} reports the values.

\begin{figure}[ht]
  \centering
  \includegraphics[width=\textwidth]{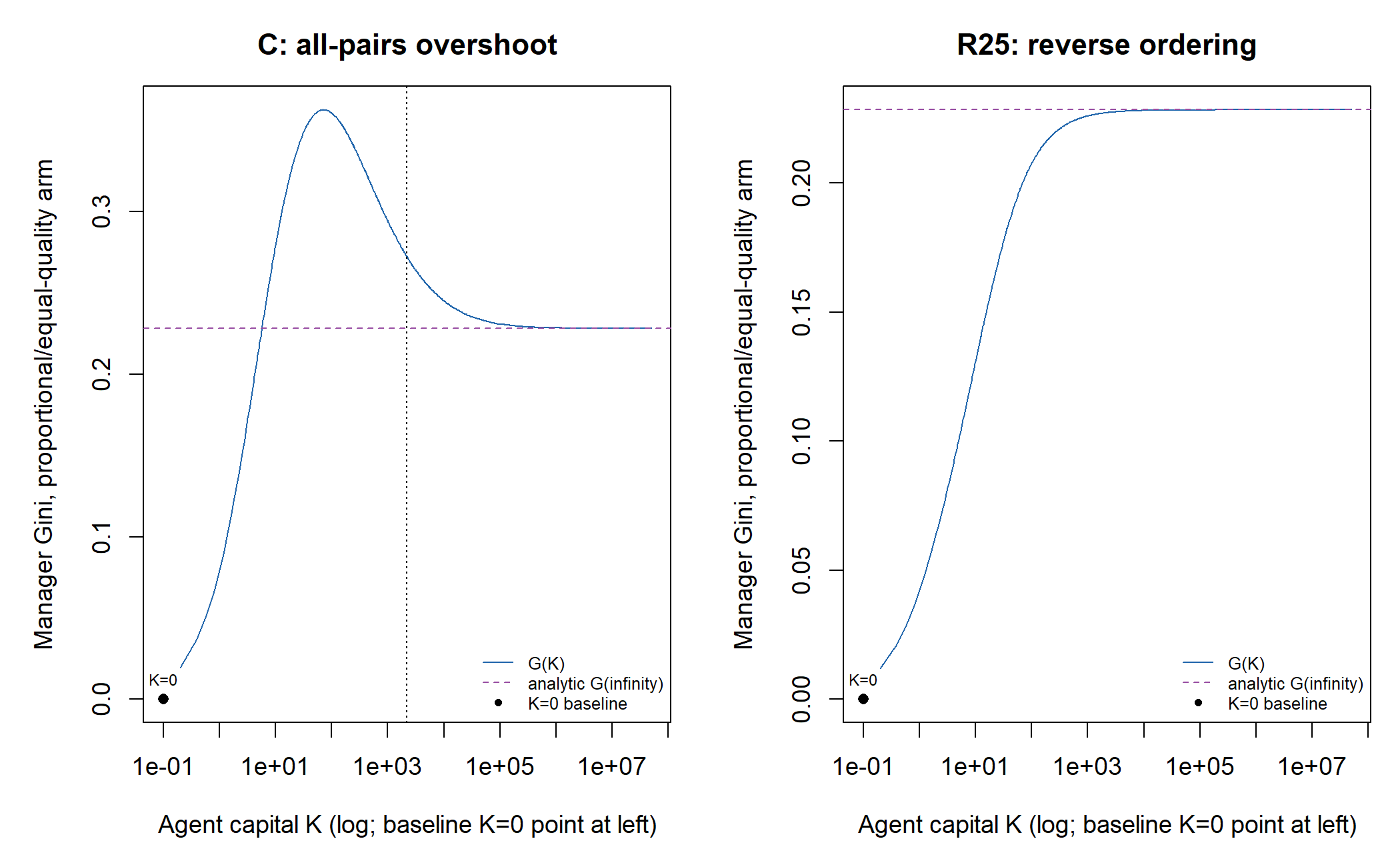}
  \caption{Manager-only Gini in the proportional-allocation exercise, where FP, PA and RS hold, $\alpha=0.65$. Left: $\beta=3$ with power floor $(f_0,\theta)=(0.005,1)$, for which condition C holds; the vertical dotted line is the largest pair crossing $\bar K$ and the dashed line is the analytic limit $G_\infty$. Right: $\beta=0.2$ with $(f_0,\theta)=(0.010,1)$, which satisfies the condition of Corollary~\ref{cor:reverse}. The horizontal axis is logarithmic and the $K_A=0$ point is marked separately.}\label{fig:floor-gini}
\end{figure}

\section{Employment and the Regime Fork}\label{sec:fork}

Coordination compression can also add tasks. Whether they are staffed depends on supervisory capacity.

\begin{proposition}[Task Frontier Expansion]\label{prop:frontier}
  Under a short-run, fixed-manager-pool and no-entry accounting closure, if $\delta>0$, $T(K_A)$ is strictly increasing and $E(K_A)=\min(\sum_i S_i(K_A),T(K_A),N)$ weakly increases. A finite capacity limit arises only if all $f_i>0$; a binding employment ceiling requires that limit to be below $N$.
\end{proposition}

Proposition~\ref{prop:frontier} is an accounting result for a fixed manager pool, not an economy-wide equilibrium. Employment is the smallest of capacity, the frontier and the workforce, and the first two rise with capital. If every floor is positive, capacity tends to $\sum_i1/f_i$, and with $\delta>0$ employment approaches $\min(N,\sum_i1/f_i)$. New tasks can then go unstaffed. If any floor is zero, capacity diverges and employment eventually reaches $N$ when $\delta>0$. With a fixed manager pool the positive-floor capacity can bind. Free managerial entry makes the floor a cost rather than a hard ceiling; hiring or promoting managers requires a separate entry and wage model.

The taxonomy in Table~\ref{tab:regimes} crosses execution complementarity ($\beta$) with task-frontier growth ($\delta$). Residual fork costs $\{f_i\}$ add a third dimension, so the regime labels describe transitional possibilities rather than four unconditional distributional predictions.\footnote{An earlier version of this paper stated, among its limitations, that a floor would leave the regime fork qualitatively unchanged. Positive floors change long-run dispersion (Appendix~\ref{app:changes}).} Employment capacity uses the fixed-pool closure FP. Under FP, PA and RS, $\beta$ and $\gamma$ govern the wage transition, while, when every floor is positive, the floors determine its long-run limit. Appendix Table~\ref{tab:registered-anchors} reports numerical values for each regime.

\begin{table}[ht]
  \centering
  \caption{Four Regimes of Coordination Compression along the Transition}\label{tab:regimes}
  \begin{tabular}{@{}lcc@{}}
    \toprule
    & Low $\delta$ (fixed task space) & High $\delta$ (task creation) \\
    \midrule
    Low $\beta$ (general infrastructure) & Gentle Compression & Rising Tide \\
    High $\beta$ (elite complementarity) & Winner Takes All & Creative Destruction \\
    \bottomrule
  \end{tabular}
\end{table}

In the descriptions below, employment statements assume FP and wage statements assume FP, PA and RS.

\begin{description}[style=nextline]
  \item[Gentle Compression (low $\beta$, low $\delta$)]
    Execution savings are broadly distributed, while a fixed task frontier and residual fork capacity can limit employment. A common positive $f$ makes managerial wage differences eventually recede.
  \item[Rising Tide (low $\beta$, high $\delta$)]
    The frontier expands, but new tasks can remain unstaffed when residual supervisory capacity binds. Broad execution access does not remove this constraint.
  \item[Winner Takes All (high $\beta$, low $\delta$)]
    Top managers may scale more rapidly at intermediate adoption levels. If $\theta>0$, skills differ and $\beta[1-f(s_{\min})/c_0]>\theta$, condition C holds and Theorem~\ref{thm:overshoot} gives population managerial-Gini overshoot and convergence from above. A common positive $f$ eventually equalizes spans; with positive floors, lasting concentration requires different floors.
  \item[Creative Destruction (high $\beta$, high $\delta$)]
    Frontier growth and transitional elite scaling can coexist. The same conditional overshoot theorem may apply, independently of $\delta$. New work may outrun fork-decision capacity under a fixed manager pool, so neither lasting concentration nor full task absorption is guaranteed.
\end{description}

\section{Numerical Illustration}\label{sec:numerical}

This section illustrates the propositions for one calibration. The floors are illustrative rather than estimated, and the numbers are not evidence about firms. We report two exercises. The \emph{proportional-allocation exercise} imposes FP, PA and RS, so the analytic limits and Theorem~\ref{thm:overshoot} apply to it. The \emph{sorted-worker simulation}, run both with zero floors (the stylized newsroom of Appendix~\ref{app:zero-figures}) and with positive floors, keeps heterogeneous workers, matches them to managers by PAM and rounds team sizes to whole numbers. It violates PA, so the closed-form wage and Gini limits do not apply to it. Table~\ref{tab:parameters} lists the parameters, and Appendix~\ref{app:numerical} gives implementation details.\footnote{For the positive-floor exercises, the evaluation points, grids and summary statistics were fixed in a pre-registration before the simulations were run. The pre-registration, its documented amendments, the code and the output files are in the replication package, which is available from the author on request.}

\begin{table}[ht]
  \centering
  \caption{Simulation parameters}\label{tab:parameters}
  \small
  \begin{tabular}{@{}lp{4.2cm}p{5.8cm}@{}}
    \toprule
    Parameter & Value & Interpretation \\
    \midrule
    \multicolumn{3}{@{}l}{\emph{Panel A. Zero-floor sorted-worker simulation ($f_i=0$)}}\\
    $A$ & 1.0 & TFP (normalized) \\
    $\alpha$ & 0.65 & Labor share in production \\
    $c_0$ & 0.3 & Baseline per-link coordination cost \\
    $\gamma$ & 1.0 & Coordination compression rate \\
    $T_0$ & 200 & Baseline task frontier \\
    $n_{\text{managers}}$ & 20 & Manager pool \\
    Manager skills & $[0.05, 1.0]$ evenly spaced & Managerial skill heterogeneity \\
    $N_{\text{workers}}$ & 400 & Worker pool \\
    Worker skills & $\text{Beta}(2, 5)$ quantiles & Right-skewed (mean 0.286) \\
    Assignment & PAM & Best workers to best managers \\
    Allocation & Largest-remainder & Discrete worker assignment \\
    $K_A$ range & $[0, 10]$ by 0.2 & Agent capital sweep \\
    $\beta_{\text{low}} / \beta_{\text{high}}$ & 0.2 / 3.0 & Infrastructure / elite complementarity \\
    $\delta_{\text{low}} / \delta_{\text{high}}$ & 0.0 / 0.3 & Fixed tasks / task creation \\
    \addlinespace
    \multicolumn{3}{@{}l}{\emph{Panel B. Positive floors (other parameters as in Panel A)}}\\
    Common floor & $f=0.06,\ 0.12$ & Special case \\
    Power floor & $f_0=0.005,\ 0.010$; $\theta=0.5,\ 1$ & $f_i=f_0s_i^{-\theta}$ \\
    Zero floor & $f=0$ & Nested benchmark \\
    Proportional allocation & $n_i=\lambda S_i$, equal mean quality & Continuous; satisfies FP, PA and RS \\
    Sorted workers & PAM or random assignment; whole-number teams & Heterogeneous workers; violates PA \\
    \bottomrule
  \end{tabular}
\end{table}

\paragraph{The overshoot case.} In the proportional-allocation exercise with $(\beta,f_0,\theta)=(3,0.005,1)$, all 190 distinct-floor pairs qualify, $G_\infty=0.228$, and the largest pair crossing is $\bar K\approx 2{,}180$. Managerial Gini is $0$ at $K_A=0$, $0.216$ at $K_A=5$ and $0.279$ at $K_A=10$. At the first grid point above $\bar K$ ($K_A\approx 2{,}220$) it is $0.272$, which exceeds its limit by $0.044$. The pair with the highest and lowest skills, $(s_i,s_j)=(1,0.05)$, illustrates Lemma~\ref{lem:crossing}: its wage ratio crosses its limit of $7.01$ at $K_A\approx 28.6$ and is $7.11$ at the next grid point, $K_A\approx 29.6$.

\paragraph{The reverse case.} With $(\beta,f_0,\theta)=(0.2,0.010,1)$, no pair qualifies, and managerial Gini stays below the same limit of $0.228$: $G(5)=0.100$, $G(10)=0.131$ and $G(2870)=0.227$, where $K_A=2870$ is the first grid point at which every manager has reached 99\% of her limiting span.

\paragraph{How far $K_A=10$ is from the long run.} At $\beta=3$ with a common floor $0.06$, the slowest manager does not reach 99\% of her own limiting span until $K_A=3.16\cdot10^6$. The plotted capital range is therefore not a saturated long run, and values at $K_A=10$ are transitional.

Appendix~\ref{app:zero-figures} reports the zero-floor sorted-worker simulation and its economy-wide patterns: between-layer inequality, the manager--worker wage gap and employment expansion. Appendix~\ref{app:numerical} reports the remaining numerical results: capacity limits under the common and power floors, two populations in which condition C fails, the sorted-worker and random-assignment simulations with positive floors, labor-share sensitivity, and values for each regime.

\section{Discussion}\label{sec:discussion}

\subsection{Predictions and Measurement}

Propositions~\ref{prop:span} and~\ref{prop:manager-demand} predict initially expanding spans and declining continuous manager demand, followed by saturation if residual decision costs remain fixed. The relevant measurements are per-supervisor decision frequency and review time by ex ante decision horizon, AI adoption and inference budget, staffing and reporting spans, task difficulty, realized outcomes, and exposure to previous attempts. A direct test of the per-link assumption is whether residual fork-review time per subordinate persists as routine execution is automated. Fork-level accuracy by horizon in \citet{pan2026} is an observable proxy for \emph{model judgment}, not a measure of $f_i$, $\beta$, human managerial ability, or firm wages. Identification would require comparable organizational workflows observed before and after execution automation and independent outcomes for consequential decisions. Sorting of harder work to more-skilled supervisors, endogenous delegation, training, and institutional memory all threaten a causal mapping from benchmark accuracy to span limits. Occupation-level exposure indices alone do not resolve these problems.

\subsection{Policy}

The model distinguishes three views of the technology, each with a different policy emphasis. If AI substitutes for labor, and coordination compression removes middle-management roles faster than new tasks absorb displaced workers (low $\delta$), transition assistance and redistribution may be necessary, although they address the consequences of compression rather than its mechanism. If AI complements workers, the emphasis is human capital investment, following \citet{agrawal2026} on expanding $\Omega(i)$ through education and multidimensional skill development; this is necessary but not sufficient when organizational structure, rather than worker capability, is the binding constraint. If AI compresses coordination, as in this paper, access and organizational design matter. Broad access to execution automation may widen spans, while unresolved fork decisions can remain a capacity constraint. Access, decision training and reusable organizational memory may affect different terms of $c_i$, so their distributional consequences cannot be inferred from $\beta$ alone.

In the regime taxonomy, $\beta$ describes how managerial skill conditions returns to \emph{execution} compression, and platform access and training may affect that mapping. Residual fork-decision costs $\{f_i\}$ are a distinct determinant of long-run supervisory capacity. The present model cannot infer whether $f_i$ is permanently scarce or how organizations price it.

\subsection{Limitations}\label{sec:limitations}

The most important limitation is that the floors are imposed. The manager-specific $f_i$ and its correlation with skill, including the power-floor gradient $\theta$, are assumed rather than estimated, and the numerical exercises illustrate the model for one calibration; they are not evidence about firms.

The evidence on floors is limited too. Firm-level hierarchy changes reported by \citet{ewens2025} and \citet{babina2023} do not identify the residual fork cost or establish that a floor binds. The model benchmark of \citet{pan2026} motivates the floor but does not identify it. It measures model judgment on mined software and research forks, not managerial decisions inside firms, and it does not report human ex-ante accuracy. Training and advice replayed from earlier runs raise model performance, so the benchmark cannot show that the floor is permanent; Appendix~\ref{app:extensions} gives the details.

The team-production and rent-sharing rule abstracts from bargaining, managerial entry, worker sorting, and general equilibrium; \citet{coase1937,williamson1979} address firm boundaries and governance outside this fixed-pool closure. The common baseline $c_i(0)=c_0$ normalizes away pre-existing managerial heterogeneity; allowing baseline differences in coordination cost or baseline sorting would introduce distributional mechanisms not analyzed here. Task creation remains a direct function of $K_A$, not a consequence derived from spare fork-decision capacity. The restriction $\beta\geq0$ is a modeling choice, and the comparative statics are silent on $\beta<0$ (Appendix~\ref{app:extensions}).

\section{Conclusion}\label{sec:conclusion}

Agent capital can reduce execution overhead while leaving some consequential fork decisions as a residual per-link input. When every manager has a positive floor, incumbent spans approach finite limits; with zero floors, spans grow without bound. Under FP, PA and RS, and when every floor is positive, those floors determine long-run managerial dispersion independently of $\beta$ and $\gamma$, which govern only the transition. Under the all-pairs condition C (every lower-floor manager approaches capacity faster), managerial inequality exceeds its long-run value at sufficiently large finite capital and converges to it from above, though not necessarily monotonically. This is a boundary to the classical communication-cost comparative static, not evidence that human judgment is permanently scarce. Employment accounting assumes a fixed manager pool and no entry; with entry, the floor is a cost rather than an economy-wide ceiling. The regime fork therefore maps transitional paths rather than unconditional long-run classes.

The floors are imposed rather than measured, and the numerical exercises are illustrations. Two questions follow. The first is empirical. How large is the residual per-link coordination cost, and how does it vary across managers? Answering it requires firm-level tests of per-link decision time and floor heterogeneity. The second is theoretical. Condition C is sufficient for overshoot, and a necessary and sufficient condition for populations in which only some pairs qualify remains open.

\bibliographystyle{apalike}
\bibliography{references}

\clearpage
\appendix
\counterwithin{figure}{section}
\counterwithin{table}{section}
\renewcommand{\thefigure}{\thesection\arabic{figure}}
\renewcommand{\thetable}{\thesection\arabic{table}}

\section{Proofs}\label{app:proofs}

We first state the Lorenz step used for Proposition~\ref{prop:wage-dispersion}, Theorem~\ref{thm:overshoot} and Corollary~\ref{cor:reverse}. For positive values $x_1\leq\dots\leq x_n$, the Lorenz shares are $L_k(x)=\sum_{i\leq k}x_i/\sum_{i\leq n}x_i$, with $L_0=0$, and the Gini coefficient is $G(x)=1-\frac1n\sum_{k=1}^n(L_{k-1}+L_k)$, which is decreasing in every $L_k$.

\begin{lemma}[Lorenz step]\label{lem:lorenz}
  Let $v_1\leq\dots\leq v_n$ be positive weights and $\tau_1,\dots,\tau_n$ positive factors, and suppose the products $v_i\tau_i$ are weakly increasing in $i$. (a) If $\tau$ is weakly increasing in $i$, then $L_k(v\tau)\leq L_k(v)$ for every $k$, strictly at every $k<n$ with $\tau_k<\tau_{k+1}$; hence $G(v\tau)\geq G(v)$, strictly if some such $k$ exists. (b) If $\tau$ is weakly decreasing in $i$, the inequalities reverse: $L_k(v\tau)\geq L_k(v)$, strictly at every $k<n$ with $\tau_k>\tau_{k+1}$, and $G(v\tau)\leq G(v)$, strictly if some such $k$ exists.
\end{lemma}

\begin{proof}
$L_k(v\tau)\leq L_k(v)$ is equivalent to $\sum_{i\leq k}v_i\tau_i/\sum_{i\leq k}v_i\leq\sum_{i\leq n}v_i\tau_i/\sum_{i\leq n}v_i$: the $v$-weighted mean of $\tau$ over the bottom $k$ is no larger than the full weighted mean. In case (a) the bottom mean is at most $\tau_k$ and the mean over the remaining indices is at least $\tau_{k+1}\geq\tau_k$. The full mean is a convex combination of the two, so it is at least the bottom mean, strictly when $\tau_k<\tau_{k+1}$. Case (b) is symmetric. Because the products are ordered, $L_k(v\tau)$ are the Lorenz shares of $v\tau$, and the Gini comparison follows from the formula above.
\end{proof}

\begin{proof}[Proof of Proposition~\ref{prop:output}]
For each nonempty team, $\partial c_i/\partial K_A=-(c_0-f_i)\gamma s_i^\beta/(1+\gamma K_A s_i^\beta)^2<0$, so $\partial L_{\text{eff},i}/\partial K_A=-Q_i n_i(\partial c_i/\partial K_A)/(1+c_i n_i)^2>0$. Team output rises because $\alpha>0$. Summing establishes the claim. Each fixed-team output converges to $A[Q_i/(1+f_i n_i)]^\alpha$, and aggregate output to their sum.
\end{proof}

\begin{proof}[Proof of Proposition~\ref{prop:span}]
Let $a_i=\gamma s_i^\beta>0$. Then $S_i=(1+a_iK_A)/(c_0+f_i a_iK_A)$ and $S_i'=(c_0-f_i)a_i/(c_0+f_i a_iK_A)^2>0$. With $f_i>0$, $S_i\to1/f_i$; with $f_i=0$, $S_i=(1+a_iK_A)/c_0$. For common positive $f$, the derivative of $S_i'$ by $a_i$ has the sign of $c_0-fa_iK_A$; the skill cross-partial can change sign. With unequal floors, skill cross-partials also depend on how $f_i$ varies with skill.
\end{proof}

\begin{proof}[Proof of Proposition~\ref{prop:manager-demand}]
For homogeneous $S=(1+\gamma K_A)/(c_0+f\gamma K_A)>1$, continuous $M=N/(S-1)$ gives $M'=-NS'/(S-1)^2<0$. Its limit is $N/(1/f-1)>0$ when $f>0$, while the $f=0$ approximation tends to zero; neither is a literal integer headcount. In a heterogeneous pool, every $S_i$ grows, so any covering set remains feasible at higher $K_A$ and the minimal cardinality cannot increase. If all floors are positive, $D=\sum_i S_i\to\sum_i1/f_i$, and manager $i$'s long-run capacity share is $(1/f_i)/\sum_j(1/f_j)$. Mixed zero floors require separate limits.
\end{proof}

\begin{proof}[Proof of Proposition~\ref{prop:wage-dispersion}]
Set $n_i=\lambda S_i$ for a common $\lambda>0$. Then $c_i n_i=\lambda$; equal positive average worker quality $q$ implies $Q_i=q\lambda S_i$ and $w_{M,i}=(1-\alpha)A[q\lambda/(1+\lambda)]^\alpha S_i^\alpha$. The common factor cancels from ratios and Gini. If all $f_i>0$, Proposition~\ref{prop:span} gives $S_i\to1/f_i$, proving the limits. With common $f>0$, all spans converge to the same value. If all $f_i=0$, $S_i=(1+\gamma K_A s_i^\beta)/c_0$; for $s_i>s_j$ and $\beta>0$, the span ratio has derivative with positive numerator $\gamma(s_i^\beta-s_j^\beta)$. Higher-ranked managers have greater proportional wage growth. Ordering skills from lowest to highest, the weighted mean growth factor of any bottom group is no larger than that of the full group, so each bottom-ranked Lorenz share falls and Gini rises (Lemma~\ref{lem:lorenz}(a), with $v$ the wages at the initial capital level and $\tau$ the growth factors). The limiting ratio is $(s_i/s_j)^{\alpha\beta}$. If only some floors are zero, those spans diverge and the positive-floor limit does not apply.
\end{proof}

\begin{proof}[Proof of Corollary~\ref{cor:common}]
With a common $f\in(0,c_0)$, $\partial S/\partial a=K_A(c_0-f)/(c_0+faK_A)^2>0$ for $K_A>0$. With $\beta>0$ and distinct skills the $a_i$ are distinct, so spans, and by Proposition~\ref{prop:wage-dispersion} wages, differ at every positive finite $K_A$, and $G(K_A)>0$. Proposition~\ref{prop:wage-dispersion} gives $G_\infty=0$.
\end{proof}

\begin{proof}[Proof of Lemma~\ref{lem:crossing}]
Since $S_i/S_j>f_j/f_i$ is equivalent to $t_i>t_j$, compare $t_i=f_i(1+a_iK_A)/(c_0+f_ia_iK_A)$ with $t_j$. Cross-multiplying, $t_i-t_j$ has the sign of
\[
  c_0(f_i-f_j)+K_A\,[f_i a_i(c_0-f_j)-f_j a_j(c_0-f_i)] .
\]
With $f_i<f_j$ the first term is negative, so $t_i>t_j$ exactly when the bracket is positive and $K_A>K^*_{ij}$. Dividing the bracket by $f_ia_if_ja_j>0$ shows that it is positive exactly when $h_j>h_i$. If $h_i\geq h_j$, the expression is negative at every $K_A\geq0$. For the power floor, $h(s)=[(c_0/f_0)s^{\theta}-1]s^{-\beta}/\gamma$, and differentiating gives the displayed $dh/ds$. Using $s^\theta/f_0=1/f(s)$, the bracket in $dh/ds$ is negative exactly when $(c_0/f(s))(\theta-\beta)+\beta<0$, that is, when $\beta[1-f(s)/c_0]>\theta$, and positive exactly when $\beta[1-f(s)/c_0]<\theta$. For $\theta>0$, $f(s)$ falls with $s$, so $1-f(s)/c_0$ rises with $s$ and the condition at $s_{\min}$ implies it throughout the support. Then $h$ is strictly decreasing, and since $f_i<f_j$ exactly when $s_i>s_j$, every pair with $f_i<f_j$ has $h_i<h_j$.
\end{proof}

\begin{proof}[Proof of Theorem~\ref{thm:overshoot}]
Put $v_i=f_i^{-\alpha}$ and $\tau_i=(f_iS_i)^\alpha$, so $w_{M,i}$ is proportional to $v_i\tau_i$. For $K_A>\bar K$, a lower floor gives both a higher $v_i$ and a higher $\tau_i$ (Lemma~\ref{lem:crossing}). Order $v_i$ increasingly, breaking ties by $\tau_i$. Both $v_i$ and $\tau_i$ then rise weakly, and $\tau_i$ rises strictly between distinct floor classes. For any bottom $k$, the $v$-weighted mean of $\tau$ in that group cannot exceed the full weighted mean. Hence every bottom Lorenz share of $v\tau$ is no larger than that of $v$, and the share at the boundary of the highest floor class (the lowest $v$ class) is strictly smaller (Lemma~\ref{lem:lorenz}(a)). Thus $G(K_A)>G_\infty$. Equal spans at zero capital imply $G(0)=0$, while distinct floors give $G_\infty>0$. Continuity and the positive-floor limits complete the proof. The power-floor condition of Lemma~\ref{lem:crossing} satisfies C when $\theta>0$ and skills differ.
\end{proof}

\begin{proof}[Proof of Corollary~\ref{cor:reverse}]
The strict support-wide condition $\beta[1-f(s_{\max})/c_0]<\theta$ gives $\beta[1-f(s)/c_0]<\theta$ throughout the support, because $1-f(s)/c_0$ rises with $s$; by Lemma~\ref{lem:crossing}, $h$ is then strictly increasing. For $K_A>0$, $S$ increases with $a(s)$ and decreases with $f(s)$, so spans and limiting wage bases $v_i=f_i^{-\alpha}$ have the same skill order. Increasing $h$ gives $t_i=f_iS_i<t_j=f_jS_j$ whenever $f_i<f_j$, at every finite $K_A$; thus $\tau_i=t_i^\alpha$ decreases across the ordered $v$ classes. Since $v_i\tau_i=S_i^\alpha$ remains ordered, every bottom Lorenz share of $v\tau$ is at least that of $v$, strictly at a floor-class boundary (Lemma~\ref{lem:lorenz}(b)). Hence the Gini remains below its limiting value, which it approaches from below. At $K_A=0$, $G(0)=0<G_\infty$.
\end{proof}

\begin{proof}[Proof of Proposition~\ref{prop:frontier}]
Since $T'=T_0\delta>0$ and $D=\sum_iS_i$ increases by Proposition~\ref{prop:span}, their minimum with fixed $N$ is weakly increasing and continuous. If every $f_i>0$, $D\to\sum_i1/f_i$, and for $\delta>0$ employment approaches $\min(N,\sum_i1/f_i)$. If any $f_i=0$, $D$ diverges and employment eventually reaches $N$ when $\delta>0$.
\end{proof}

\section{Extensions and Remarks}\label{app:extensions}

\begin{remark}[Per-manager obligations]
If a fixed project-level obligation uses $F_i\in[0,1)$ of the manager's unit time separately from the per-link coordination cost, $F_i+c_i n_i\leq1$ gives $S_i=(1-F_i)/c_i$. Its limit is $(1-F_i)/f_i$ for $f_i>0$ and is unbounded for $f_i=0$. Thus a fixed obligation alone cannot create a span ceiling. Under $n_i=\lambda S_i$ with equal positive mean team quality, $c_i n_i=\lambda(1-F_i)$, so $w_{M,i}\propto[S_i/(1+\lambda(1-F_i))]^\alpha$. If $\lambda\to\lambda_\infty>0$, the long-run wage base is $\widetilde v_i=[(1-F_i)/(f_i(1+\lambda_\infty(1-F_i)))]^\alpha$. At binding capacity ($\lambda_\infty=1$), the $i/j$ ratio is $[(1-F_i)f_j(2-F_j)/((1-F_j)f_i(2-F_i))]^\alpha$; for common scarce $\lambda_\infty<1$ the manager-specific denominators remain. Heterogeneous $F_i$ also makes baseline wages unequal. Although $f_iS_i/(1-F_i)$ follows the original convergence path, condition C alone does not imply Gini overshoot when the wage-base rankings or their time-varying denominators differ. Some heterogeneous-$F_i$ populations can still overshoot.
\end{remark}

\paragraph{Mixed zero and positive floors.} The long-run results of Proposition~\ref{prop:wage-dispersion} and Theorem~\ref{thm:overshoot} require every floor to be positive. If only some floors are zero, those spans diverge and the positive-floor limit does not apply; total capacity diverges, and with $\delta>0$ employment eventually reaches $N$ (Proposition~\ref{prop:frontier}). Mixed zero floors require separate limits, which we do not derive.

\paragraph{Allocation and output.} Under any allocation mechanism that is at least weakly improving, choosing among allocations at least as good as the status quo, Proposition~\ref{prop:output} extends to endogenous reallocation. The allocation rule of the sorted-worker simulation is not globally optimal and can produce local output non-monotonicities when expanding the workforce admits low-skill workers whose coordination overhead exceeds their quality contribution. While the proof establishes $\partial Y / \partial K_A > 0$ holding team allocations fixed, the sorted-worker simulation allocates workers in proportion to capacity from an expanding workforce. When the task frontier grows ($\delta > 0$), newly admitted low-skill workers can temporarily add more coordination friction than quality, producing local output decreases along the $K_A$ path. This is a feature of the allocation rule, not a violation of the proposition: at any fixed $K_A$, output is higher than it would be at lower $K_A$ with the same team assignments. At full employment the extensive margin disappears. The simulated fall in economy-wide Gini accompanies an expanding employed share; it need not hold with positive floors.

\paragraph{Evidence on the floor from \citet{pan2026}.} Decisions that can be run in parallel until evidence arrives do not belong in a noncompressible floor; the residual applies only when parallel trials are irreversible, prohibitively costly, or too slow. Their benchmark, Taste-Bench, evaluates model judgments on mined software and research forks, not managerial decisions inside firms, and its detour labels partly reflect which path the agent that produced the run abandoned. Its Appendix F gathered human choices \emph{before} outcomes but reports only post-outcome judgments, whose agreement with the labels is 98.8\% among explicit preferences; human ex-ante accuracy is not reported. The two-model reasoning-budget comparison does not show that the floor is permanent: reasoning distillation of one base model improves held-out engineering-fork accuracy by 17.9 percentage points, and advice from earlier runs on the same 41 held-out tasks raises executor success from 14.6\% to 33.7\% (39.0\% with correct advice), which is closer to institutional memory than to zero-shot foresight. The four horizon bins are retrospective labels, and later evidence might point to common informational uncertainty ($\theta$ near zero) rather than differential judgment ($\theta>0$), or to accountability. The benchmark identifies neither a firm's time cost nor human superiority.

\paragraph{Negative $\beta$.} A natural question is whether $\beta$ could be negative, so execution tools particularly help less-experienced coordinators. Evidence on AI assistance \citep{brynjolfsson2025} and theory of AI in knowledge hierarchies \citep{ide2025} motivate that question; neither identifies its sign in this model. The current comparative statics are silent on $\beta<0$; future analysis requires a positive lower bound on skill and a separately specified distribution of $f_i$.

\section{Numerical Implementation}\label{app:numerical}

\paragraph{Zero-floor simulation.} The zero-floor newsroom simulation has 20 managers and $f_i=0$ (seed 2026). Each manager receives $K_{\text{eff},i}=K_A s_i^\beta$, $c_i=c_0/(1+\gamma K_{\text{eff},i})$, and span $S_i=1/c_i$. Workers are allocated by largest-remainder integer assignment and matched to managers by PAM. At each $K_A$, the employable pool is bounded by $\min(N_{\text{workers}},T(K_A))$, and assignment is further limited by capacity.

\paragraph{Positive floors.} The positive-floor exercises keep the same manager skills, deterministic worker-skill quantiles, $c_0=0.3$, $\gamma=1$, and $K_A\in[0,10]$. They compare common $f=0.06,0.12$ with $f_i=f_0s_i^{-\theta}$ for $f_0=0.005,0.010$ and $\theta=0.5,1$, with $f=0$ as the nested benchmark. The proportional-allocation exercise uses continuous $n_i=\lambda S_i$, $\lambda=\min(1,N/\sum S_i,T/\sum S_i)$, equal positive mean worker quality, and the stated wage rule. The sorted-worker simulation and its random-assignment variant keep heterogeneous workers and do not satisfy PA. A sensitivity grid varies $\beta\in\{0.2,0.5,1,2,3\}$, $\delta\in\{0,0.15,0.3,0.6\}$ and $\alpha\in\{0.50,0.65,0.80\}$. To check convergence, capital values extend to $5\cdot10^7$; these are model units, not adoption forecasts. With $f=0$, the positive-floor code reproduces the zero-floor simulation's figure arrays and worker assignments in every regime-grid cell.

\paragraph{Two implementations of integer assignment.} Figure~\ref{fig:economy-gini} in Appendix~\ref{app:zero-figures} uses the implementation that generates the figures, whose baseline values are reported there. A second, historical implementation used for verification reports economy Gini 0.824 at $K_A=0$, 0.431 for Rising Tide at $K_A=10$, manager Gini 0.401 in Creative Destruction at $K_A=10$ and manager Gini 0.0705 at its baseline; Table~\ref{tab:robustness} in Appendix~\ref{app:zero-figures} comes from this implementation. These values refer to different implementations, populations or regimes and are not a decomposition of the overall Gini. Across the plotted zero-floor paths, rising employment is associated with falling economy-wide inequality; that regularity need not persist under positive floors.

\paragraph{Capacity limits.} For a common floor $0.06$ the twenty-manager limiting capacity is about $333$, below $N=400$ but above the fixed frontier $T_0=200$; at $0.12$ it is about $167$, below both. Under the power floor $(f_0,\theta)=(0.005,1)$, the minimum, middle and maximum skill positions have limiting spans $10$, $105$ (mean of positions 10 and 11), and $200$; total limiting capacity is $2100$, above both $N$ and $T_0$, so the floor does not create a binding employment ceiling in this example.

\paragraph{Populations in which condition C fails.} In an uncalibrated three-manager example, one of three distinct-floor pairs qualifies; in a constructed 20-manager population with mixed ordering, 105 of 190 qualify. In each, at least one pair crosses its limit while the population Gini remains below its limit at every point of the finite grid. This does not show that C is necessary, and it says nothing about mixed orderings beyond the grid. The constructed population was selected after an exploratory search for feasible cases, so neither example is independent evidence.

\paragraph{Sorted workers and random assignment.} The sorted-worker simulation violates PA: whole-number assignment and PAM give teams different mean quality, so its baseline managerial Gini is positive, and even a common positive floor need not eliminate its wage dispersion. Closed-form wage and Gini limits therefore apply only to the proportional-allocation exercise. Figure~\ref{fig:floor-pam} shows the sorted-worker simulation with positive floors. In the zero-floor simulation, all five qualitative checks (output rising, manager Gini ordering, economy-wide Gini falling, gap rising, employment ordering) also hold under random worker assignment across the four regimes. This check does not extend Proposition~\ref{prop:wage-dispersion} to PAM or to positive floors.

\begin{figure}[ht]
  \centering
  \includegraphics[width=.85\textwidth]{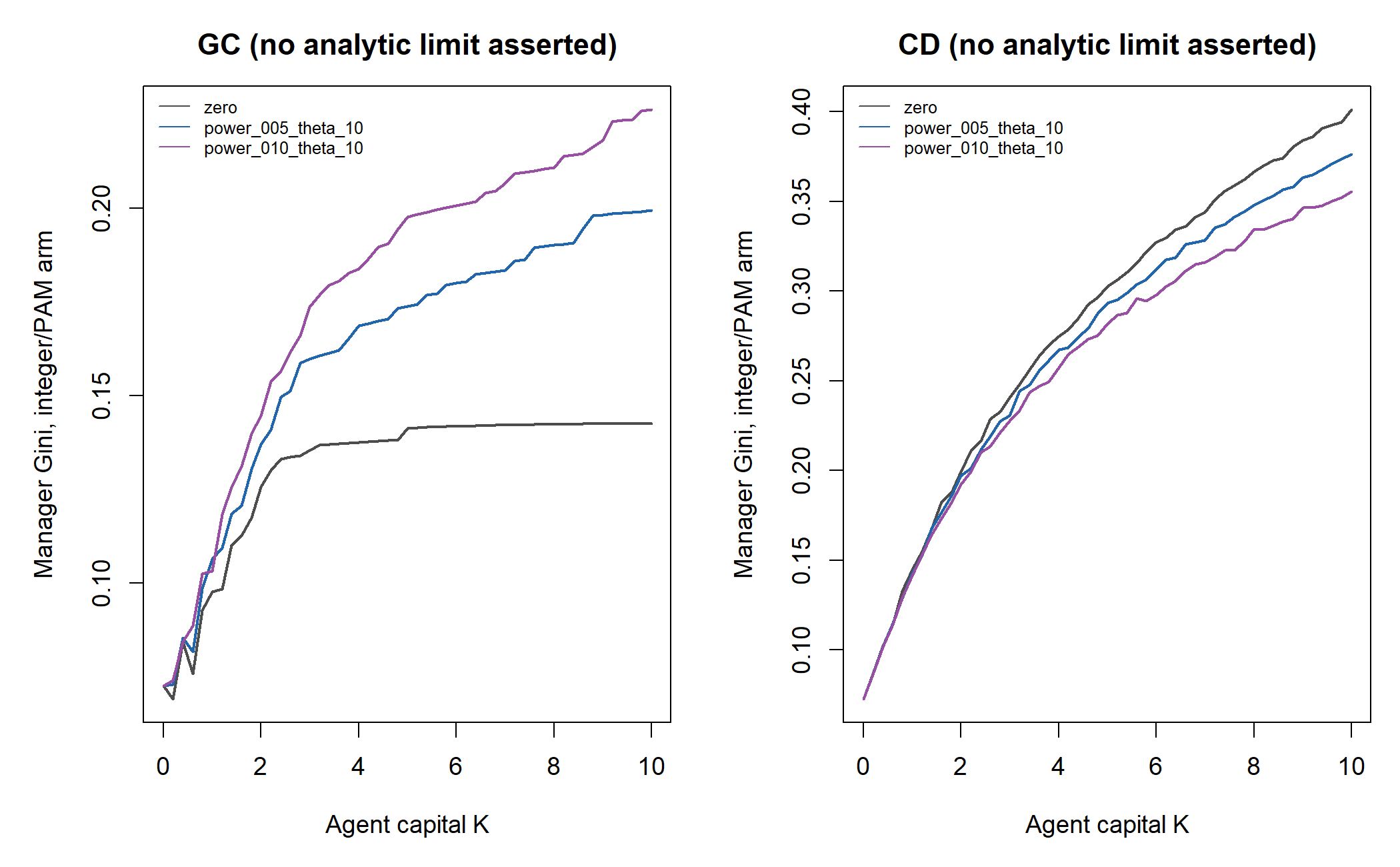}
  \caption{Sorted-worker simulation with positive floors: manager Gini for two regimes (Gentle Compression and Creative Destruction), comparing the zero floor with power floors $(f_0,\theta)=(0.005,1)$ and $(0.010,1)$ over $K_A\in[0,10]$. Because mean worker quality differs across teams, the limits of Proposition~\ref{prop:wage-dispersion} and condition C do not apply.}\label{fig:floor-pam}
\end{figure}

\paragraph{Labor share.} The zero-floor simulation uses $\alpha=0.65$ as its baseline labor share; Table~\ref{tab:robustness} reports its results for $\alpha\in\{0.50,0.65,0.80\}$, which bear only on zero floors. In the proportional-allocation exercise, Table~\ref{tab:floor-alpha} illustrates that positive-floor limits do not depend on $\beta$ for a fixed floor distribution. The power-floor limits vary with $\alpha$, taking the values $0.184$, $0.228$ and $0.269$ for $\alpha=0.50$, $0.65$ and $0.80$, but are identical across $\beta$ at each $\alpha$. The full sensitivity grid is in the replication package. All five zero-floor checks hold at each of the three values of $\alpha$.

\begin{table}[ht]
  \centering
  \caption{Manager Gini with positive floors in the proportional-allocation exercise. Each cell reports $G(10)$; the last column is the analytic limit $G_\infty$. The common floor has a zero limit for either $\beta$. Power floor: $f_0=0.005$, $\theta=1$. No ordering in $\beta$ at finite $K_A$ is claimed.}\label{tab:floor-alpha}
  \small
  \begin{tabular}{@{}llrrrr@{}}
    \toprule
    Floor & $\beta$ & $\alpha=.50$ & $\alpha=.65$ & $\alpha=.80$ & $G_\infty$ ($\alpha=.65$)\\
    \midrule
    Common $0.06$ & $0.2$ & .012 & .015 & .018 & 0\\
    Common $0.06$ & $3$ & .130 & .167 & .204 & 0\\
    Power $0.005$ & $0.2$ & .082 & .105 & .125 & .228\\
    Power $0.005$ & $3$ & .217 & .279 & .338 & .228\\
    \bottomrule
  \end{tabular}
\end{table}

\paragraph{Regime values.} Table~\ref{tab:registered-anchors} reports capacity, employment and managerial Gini for each floor and regime.

\begin{table}[p]
  \centering
  \caption{Regime values at $\alpha=0.65$. For each floor, regime and capital level: continuous capacity $D$ and employment $E$, managerial Gini $G$ in the proportional-allocation exercise, and employment $n$ and managerial Gini $G_{\rm PAM}$ in the sorted-worker simulation (figure implementation). At $K_A=0$ every row has $D=E=66.7$, $G=0.000$, $n=67$ and $G_{\rm PAM}=0.073$. $T(10)=800$ in the high-$\delta$ regimes and $200$ otherwise, and $E(10)=\min(400,T(10),D(10))$. Floors marked P are power floors $f_0s_i^{-\theta}$. GC: Gentle Compression; RT: Rising Tide; WTA: Winner Takes All; CD: Creative Destruction.}\label{tab:registered-anchors}
  \scriptsize
  \setlength{\tabcolsep}{4pt}
  \begin{tabular}{@{}ll*{10}{r}@{}}
    \toprule
    & & \multicolumn{5}{c}{$K_A=5$} & \multicolumn{5}{c}{$K_A=10$}\\
    \cmidrule(lr){3-7}\cmidrule(l){8-12}
    Floor & Regime & $D$ & $E$ & $G$ & $n$ & $G_{\rm PAM}$ & $D$ & $E$ & $G$ & $n$ & $G_{\rm PAM}$\\
    \midrule
    Zero & GC & 349.6 & 200.0 & .043 & 200 & .141 & 632.5 & 200.0 & .048 & 200 & .143\\
     & RT & 349.6 & 349.6 & .043 & 350 & .212 & 632.5 & 400.0 & .048 & 400 & .267\\
     & WTA & 158.5 & 158.5 & .225 & 159 & .303 & 250.4 & 200.0 & .296 & 200 & .378\\
     & CD & 158.5 & 158.5 & .225 & 159 & .303 & 250.4 & 250.4 & .296 & 250 & .401\\
    \addlinespace
    Common .06 & GC & 188.4 & 188.4 & .019 & 188 & .115 & 233.6 & 200.0 & .015 & 200 & .111\\
     & RT & 188.4 & 188.4 & .019 & 188 & .115 & 233.6 & 233.6 & .015 & 234 & .126\\
     & WTA & 113.9 & 113.9 & .142 & 114 & .210 & 138.3 & 138.3 & .167 & 138 & .241\\
     & CD & 113.9 & 113.9 & .142 & 114 & .210 & 138.3 & 138.3 & .167 & 138 & .241\\
    \addlinespace
    Common .12 & GC & 129.3 & 129.3 & .010 & 129 & .091 & 143.6 & 143.6 & .007 & 144 & .089\\
     & RT & 129.3 & 129.3 & .010 & 129 & .091 & 143.6 & 143.6 & .007 & 144 & .089\\
     & WTA & 93.5 & 93.5 & .091 & 94 & .155 & 103.9 & 103.9 & .103 & 104 & .167\\
     & CD & 93.5 & 93.5 & .091 & 94 & .155 & 103.9 & 103.9 & .103 & 104 & .167\\
    \addlinespace
    P .005, $\theta=.5$ & GC & 315.6 & 200.0 & .051 & 200 & .148 & 521.1 & 200.0 & .062 & 200 & .159\\
     & RT & 315.6 & 315.6 & .051 & 316 & .199 & 521.1 & 400.0 & .062 & 400 & .282\\
     & WTA & 152.2 & 152.2 & .216 & 152 & .293 & 229.5 & 200.0 & .279 & 200 & .362\\
     & CD & 152.2 & 152.2 & .216 & 152 & .293 & 229.5 & 229.5 & .279 & 230 & .375\\
    \addlinespace
    P .005, $\theta=1$ & GC & 297.1 & 200.0 & .079 & 200 & .174 & 474.4 & 200.0 & .105 & 200 & .199\\
     & RT & 297.1 & 297.1 & .079 & 297 & .216 & 474.4 & 400.0 & .105 & 400 & .325\\
     & WTA & 151.6 & 151.6 & .216 & 152 & .293 & 227.9 & 200.0 & .279 & 200 & .364\\
     & CD & 151.6 & 151.6 & .216 & 152 & .293 & 227.9 & 227.9 & .279 & 228 & .376\\
    \addlinespace
    P .010, $\theta=.5$ & GC & 287.9 & 200.0 & .058 & 200 & .156 & 444.0 & 200.0 & .072 & 200 & .168\\
     & RT & 287.9 & 287.9 & .058 & 288 & .192 & 444.0 & 400.0 & .072 & 400 & .291\\
     & WTA & 146.5 & 146.5 & .207 & 147 & .283 & 212.7 & 200.0 & .264 & 200 & .351\\
     & CD & 146.5 & 146.5 & .207 & 147 & .283 & 212.7 & 212.7 & .264 & 213 & .357\\
    \addlinespace
    P .010, $\theta=1$ & GC & 261.6 & 200.0 & .100 & 200 & .198 & 387.6 & 200.0 & .131 & 200 & .226\\
     & RT & 261.6 & 261.6 & .100 & 262 & .219 & 387.6 & 387.6 & .131 & 388 & .331\\
     & WTA & 145.5 & 145.5 & .207 & 145 & .282 & 209.8 & 200.0 & .264 & 200 & .351\\
     & CD & 145.5 & 145.5 & .207 & 145 & .282 & 209.8 & 209.8 & .264 & 210 & .356\\
    \bottomrule
  \end{tabular}
\end{table}

\section{Zero-Floor Simulation}\label{app:zero-figures}

This appendix reports the zero-floor sorted-worker simulation, the stylized newsroom of the introduction. Because it violates PA, the wage and Gini limits of Proposition~\ref{prop:wage-dispersion} do not apply to it, and its patterns need not carry over to positive floors. Figure~\ref{fig:economy-gini} compares full-economy and manager-only Gini across the four regimes, Figure~\ref{fig:heatmap} shows a transitional $(\beta,\delta)$ surface at $K_A=5$, Figure~\ref{fig:wage-distributions} shows wage densities, and Figure~\ref{fig:dashboard} shows a six-metric dashboard.

\paragraph{The zero-floor newsroom.} In the figure-generating implementation, the manager Gini at $K_A=0$ is 0.0725 (67 employed) and the economy Gini is 0.822. Three patterns appear. First, between-layer inequality dominates within-layer inequality. The manager--worker divide is the main source of inequality, and it widens as managers lead larger teams. Second, the manager--worker wage gap rises in all four simulated regimes; the mean manager earns about $1.8\times$ the mean employed worker at $K_A=0$ and about $5.4\times$ at $K_A=10$ under Gentle Compression. This follows in part from the imposed Cobb--Douglas sharing convention. Third, employment expansion is the main equalizing force. About 67 of 400 workers are employed at $K_A=0$, and under Rising Tide unemployment falls from over 80\% to near zero. With positive floors, the frontier can keep growing after capacity saturates, so neither this employment path nor the simulated Gini path need carry over.

\begin{figure}[ht]
  \centering
  \includegraphics[width=\textwidth]{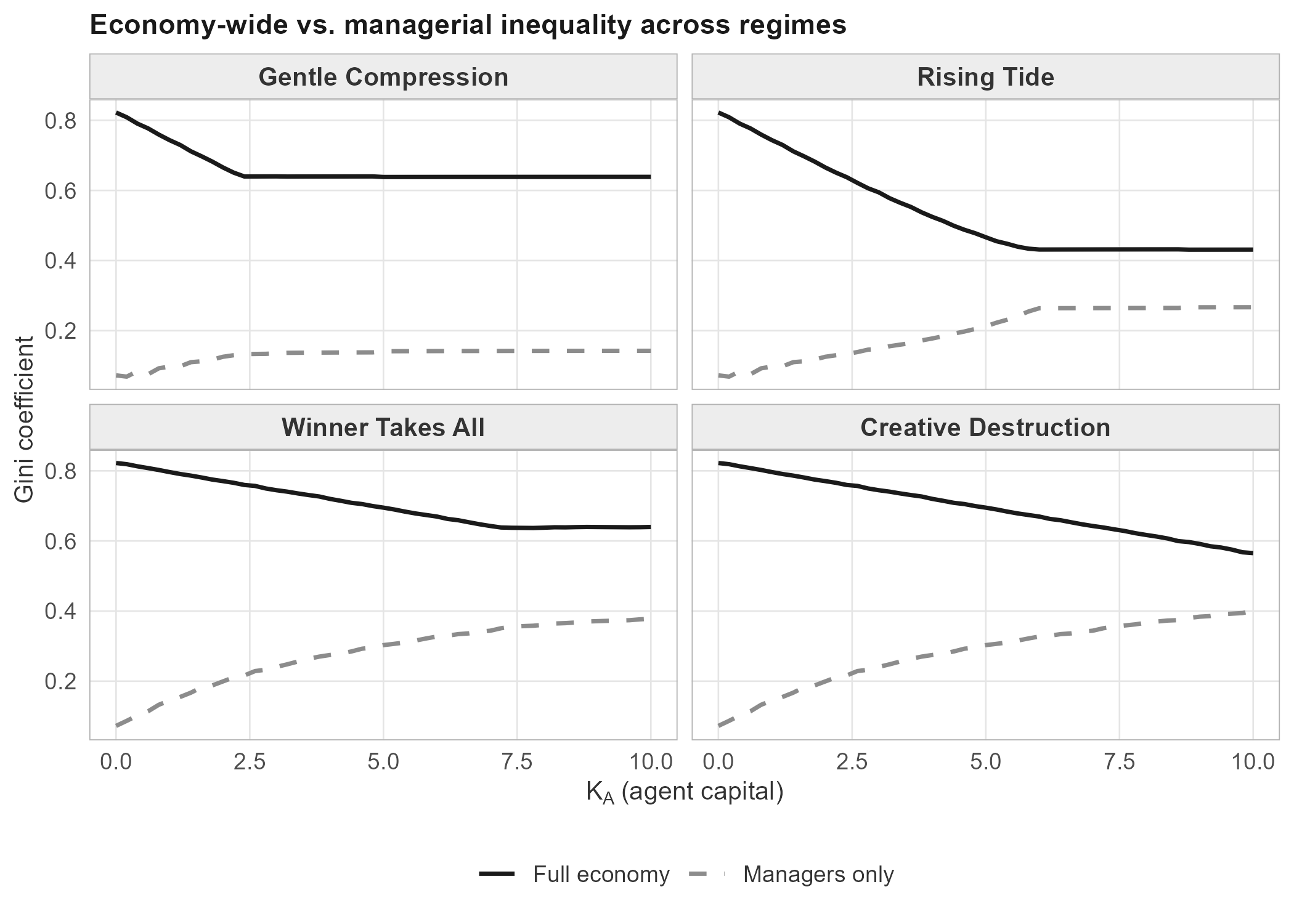}
  \caption{Zero-floor sorted-worker simulation: full-economy Gini (solid) versus manager-only Gini (dashed) across the four regimes.}\label{fig:economy-gini}
\end{figure}

\begin{figure}[ht]
  \centering
  \includegraphics[width=\textwidth]{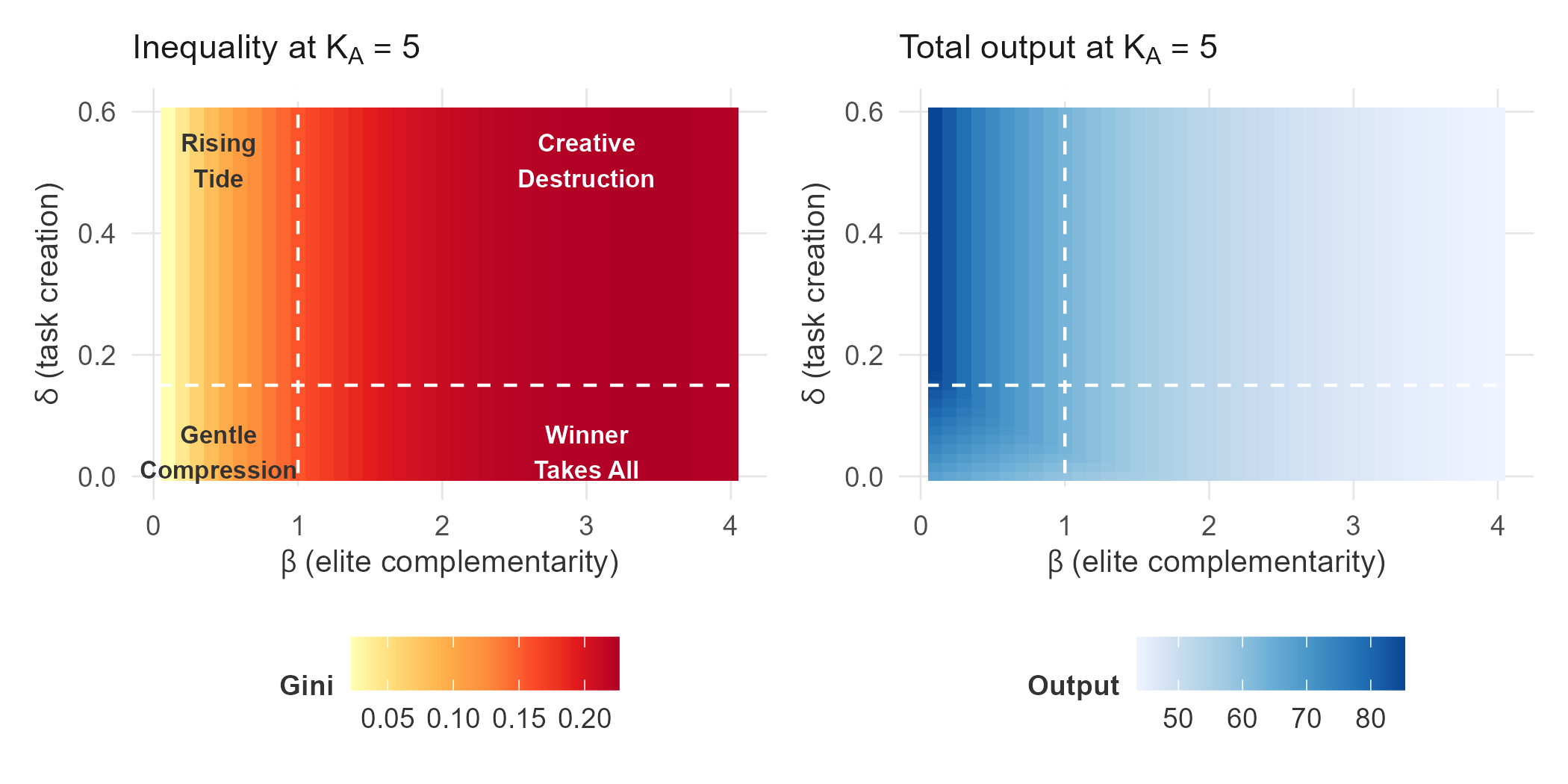}
  \caption{Zero-floor sorted-worker simulation at $K_A=5$: manager Gini and output across $(\beta,\delta)$.}\label{fig:heatmap}
\end{figure}

\begin{figure}[ht]
  \centering
  \includegraphics[width=\textwidth]{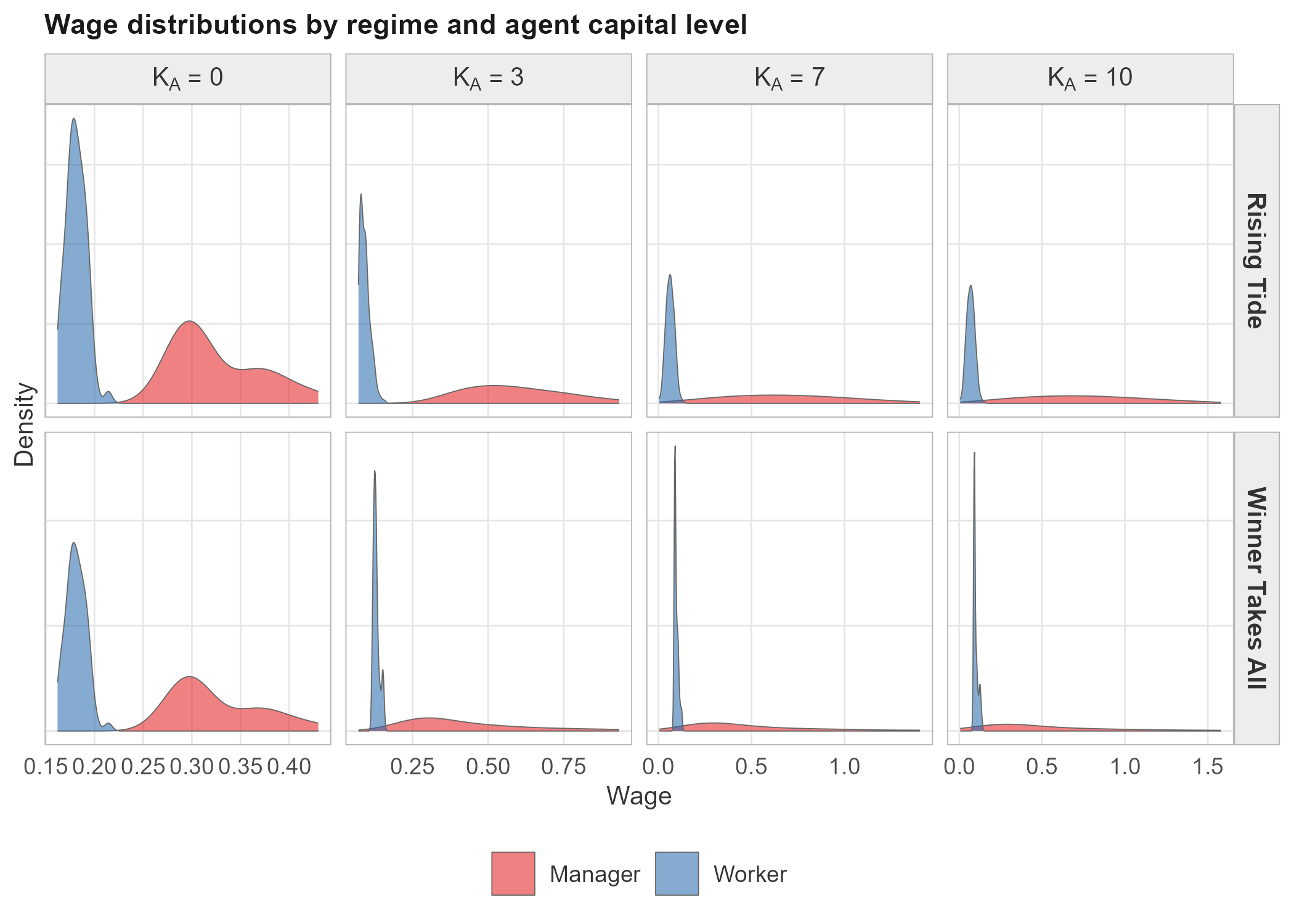}
  \caption{Zero-floor sorted-worker simulation: wage densities at $K_A=0,3,7,10$ in Rising Tide and Winner Takes All.}\label{fig:wage-distributions}
\end{figure}

\begin{figure}[ht]
  \centering
  \includegraphics[width=\textwidth, height=0.7\textheight, keepaspectratio]{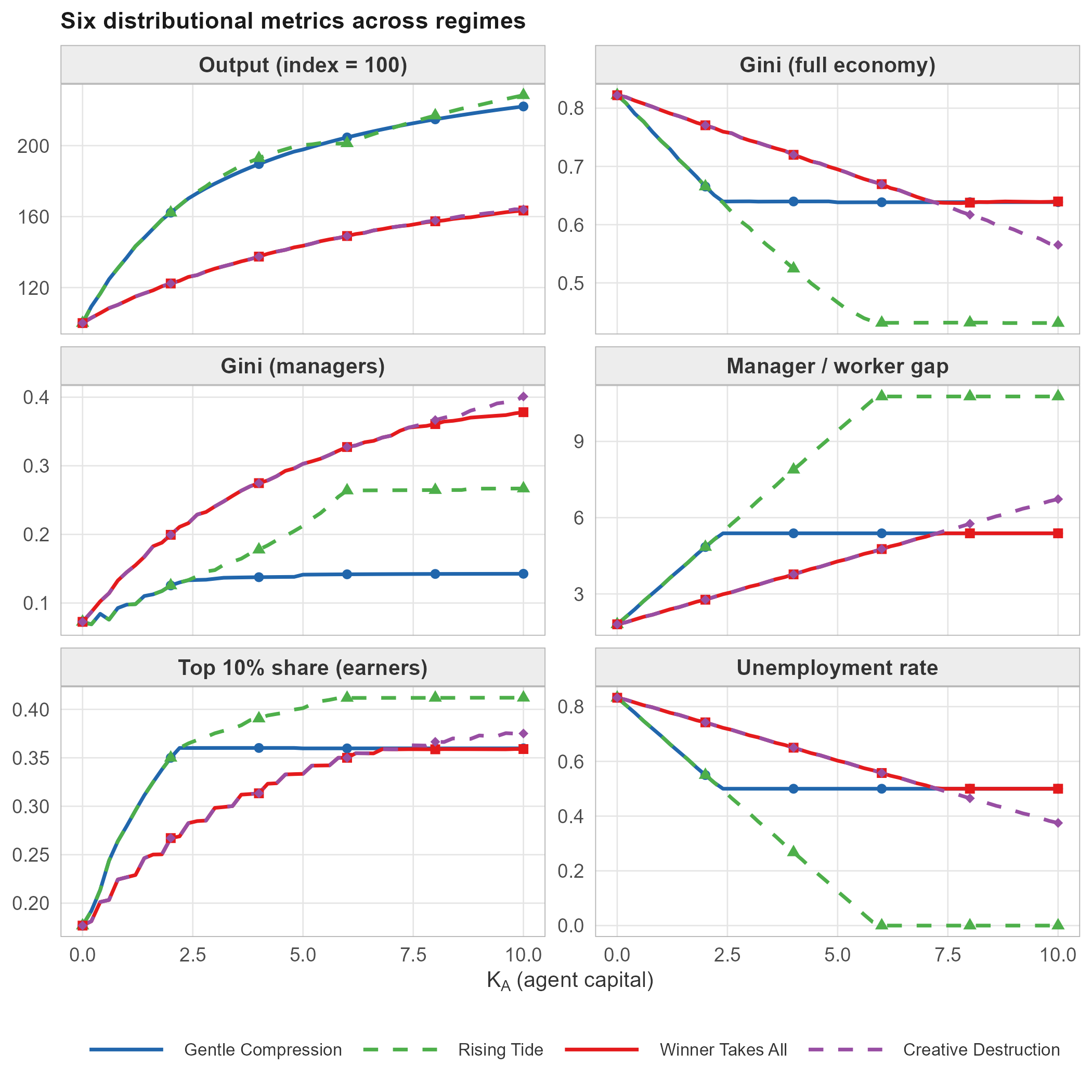}
  \caption{Zero-floor sorted-worker simulation: six-metric dashboard.}\label{fig:dashboard}
\end{figure}

\begin{table}[ht]
  \centering
  \caption{Labor-share robustness of the zero-floor sorted-worker simulation. Employment is identical across $\alpha$ because spans and the frontier do not depend on the labor share.}\label{tab:robustness}
  \begin{tabular}{@{}clccccc@{}}
    \toprule
    $\alpha$ & Metric & GC & RT & WTA & CD & Ordering \\
    \midrule
    0.50 & Mgr Gini & 0.110 & 0.213 & 0.295 & 0.315 & WTA $>$ GC, CD $>$ RT \\
    0.50 & Gap ($\times$) & 10.0 & 20.0 & 10.0 & 12.5 & All $>$ baseline \\
    0.50 & Employed & 200 & 400 & 200 & 250 & RT $\geq$ GC, CD $\geq$ WTA \\
    \addlinespace
    0.65 & Mgr Gini & 0.143 & 0.267 & 0.378 & 0.401 & WTA $>$ GC, CD $>$ RT \\
    0.65 & Gap ($\times$) & 5.4 & 10.8 & 5.4 & 6.7 & All $>$ baseline \\
    0.65 & Employed & 200 & 400 & 200 & 250 & RT $\geq$ GC, CD $\geq$ WTA \\
    \addlinespace
    0.80 & Mgr Gini & 0.175 & 0.316 & 0.454 & 0.479 & WTA $>$ GC, CD $>$ RT \\
    0.80 & Gap ($\times$) & 2.5 & 5.0 & 2.5 & 3.1 & All $>$ baseline \\
    0.80 & Employed & 200 & 400 & 200 & 250 & RT $\geq$ GC, CD $\geq$ WTA \\
    \bottomrule
  \end{tabular}
\end{table}

\clearpage
\section{Changes from the Previous Version}\label{app:changes}

This appendix lists the changes from the previous version, arXiv:2602.16078v3. Items 1--8 correct statements in v3 and do not depend on the floor; item 9 describes the extension to a floor.

\paragraph{Corrections independent of the floor.}
\begin{enumerate}[leftmargin=*,itemsep=2pt]
  \item Proposition 4 of v3 stated that the rate at which the managerial Gini rises is increasing in $\beta$. This does not hold in general: with two skills $(0.25,0.5)$ at $f=0$, the Gini's initial slope is proportional to $0.5^\beta-0.25^\beta$, which decreases in $\beta$ for $\beta>1$.
  \item The top-to-average span gap is linear, not accelerating, in capital at $f=0$.
  \item Faster convergence for higher $\beta$ is not general.
  \item $\beta=1$ is not a regime threshold.
  \item Integer manager and layer counts can remain flat between thresholds.
  \item Model-based equal baseline wages do not describe the sorted-worker simulation's nonzero baseline manager Gini (0.0725 in the implementation that generates the figures; 0.0705 in the verification implementation).
  \item Continuous proportional allocation is now capped by capacity, which the v3 text did not require; the largest-remainder allocator of the zero-floor simulation rounds ideal spans and can exceed an individual span by less than one.
  \item The top-manager capacity share rises for any $\beta>0$, not only for $\beta>1$.
\end{enumerate}

\paragraph{Substantive floor changes.}
\begin{enumerate}[leftmargin=*,itemsep=2pt,resume]
  \item The model nests v3 exactly at zero floors. The sixth limitation listed in v3 stated that a floor would leave the regime fork qualitatively unchanged. Positive floors instead change long-run dispersion, and under the all-pairs condition C managerial Gini rises above its long-run limit and converges to it from above, a non-monotone path unlike the monotone zero-floor path.
\end{enumerate}

\end{document}